\documentclass{amsart}
\usepackage[utf8]{inputenc}
\usepackage{amssymb}
\usepackage{amsthm}
\usepackage{amsfonts}
\usepackage{amsmath}
\usepackage{hyperref}
\usepackage{multicol,multirow}
\usepackage[text={440pt,575pt},headheight=9pt,centering]{geometry}
\usepackage{comment}
\usepackage{enumitem}
\usepackage[font=small]{subfig}
\usepackage{float}
\newcommand{\vertiii}[1]{{\left\vert\kern-0.25ex\left\vert\kern-0.25ex\left\vert #1 
    \right\vert\kern-0.25ex\right\vert\kern-0.25ex\right\vert}_\infty}

\newcommand{\R}[0]{\mathbb{R}}
\newcommand{\norm}[1]{\|#1\|}
\DeclareMathOperator*{\argmin}{argmin}
\usepackage{tikz}
\usetikzlibrary{shapes, positioning}

\newtheorem{theorem}{Theorem}
\newtheorem{lemma}[theorem]{Lemma}
\newtheorem{proposition}[theorem]{Proposition}

\newtheorem{assumption}{Assumption}
\newtheorem{remark}{Remark}

\usepackage{graphicx}
\usepackage{placeins}
\usepackage[comma, sort&compress]{natbib} 

\def\gsp{\textsc{gsp}}
\def\ges{\textsc{ges}}
\def\pc{\textsc{pc}}

\usepackage{color}

\def\T{{ \mathrm{\scriptscriptstyle T} }}

\def\inv{{ \mathrm{\scriptscriptstyle -1} }}
\def\invT{{ \mathrm{\scriptscriptstyle -T} }}
\def\gap{\textsc{GAP}}
\def\rgap{\textsc{RGAP}}
\def\hdbu{\textsc{HD-BU}}
\def\hdtd{\textsc{HU-TD}}
\def\td{\textsc{TD}}
\def\misocp{\textsc{MISOCP}}
\def\micp{\textsc{Our Method}}
\usepackage[plain,noend]{algorithm2e}

 \usepackage[foot]{amsaddr}

\title[Integer Programming for Bayesian Networks]{Integer Programming for Learning Directed Acyclic Graphs
from Non-identifiable Gaussian Models}
\author[T.~Xu]{Tong Xu$^1$$^\star$}
\email{tongxu2027@u.northwestern.edu}
\address{$^1$Department of Industrial Engineering and Management Sciences, Northwestern University, U.S.}
\author[A.~Taeb]{Armeen Taeb$^{2}$$^\star$}
\email{ataeb@uw.edu}
\address{$^2$Department of Statistics, University of Washington, U.S.}
\author[S.~Küçükyavuz]{Simge Küçükyavuz$^1$}
\email{simge@northwestern.edu}
\author[A.~Shojaie]{Ali Shojaie$^{3}$}
\email{ashojaie@uw.edu}
\address{$^3$Department of Biostatistics, University of Washington, U.S.}
\address{$^\star$Equal contributions}

\usepackage{fancyhdr}

\date{\today}
\usepackage{etoolbox}

\begin{document}
\maketitle

\begin{abstract}
We study the problem of learning directed acyclic graphs from continuous observational data, generated according to a linear Gaussian structural equation model. State-of-the-art structure learning methods for this setting have at least one of the following shortcomings: i) they cannot provide optimality guarantees and can suffer from learning sub-optimal models; ii) they rely on the stringent assumption that the noise is homoscedastic, and hence the underlying model is fully identifiable.  We overcome these shortcomings and develop a computationally efficient mixed-integer programming framework for learning medium-sized problems that accounts for arbitrary heteroscedastic noise. We present an early stopping criterion under which we can terminate the branch-and-bound procedure to achieve an asymptotically optimal solution and establish the consistency of this approximate solution. In addition, we show via numerical experiments that our method outperforms state-of-the-art algorithms and is robust to noise heteroscedasticity, whereas the performance of some competing methods deteriorates under strong violations of the identifiability assumption. The software implementation of our method is available as the Python package \emph{micodag}.
\end{abstract}
\section{Introduction}


\subsection{Background and {R}elated {W}orks}
A Bayesian network is a probabilistic graphical model that represents conditional independencies among a set of random variables. A directed acyclic graph serves as the structure to encode these conditional independencies, where the random variables are represented as vertices, and a pair of vertices that are not connected in a path are conditionally independent. A directed edge from node $i$ to node $j$ indicates that $i$ causes $j$. The acyclic property of the graph prevents the occurrence of circular dependencies and allows for probabilistic inference and learning algorithms for Bayesian networks. At a high level, the focus of this paper is developing methods to efficiently learn provably optimal directed acyclic graphs from continuous observational data. Throughout, we assume that all relevant variables are observed. 

Various methods exist for learning directed acyclic graphs, with the majority falling into two categories: constraint-based methods, and score-based methods. 
Constraint-based methods aim to identify conditional independencies from the data. An example is the PC algorithm \citep{causalitybase}, which initiates with a complete undirected graph and iteratively removes edges based on conditional independence assessments. \citet{kalisch2005estimating} shows that under a condition known as `strong faithfulness,' the PC algorithm is consistent in high-dimensional settings for learning sparse Gaussian directed acyclic graphs. Score-based methods, which is the approach considered in this paper, often employ penalized log-likelihood as a score function to seek the optimal graph within the entire space of directed acyclic graphs \citep{vgBuhlmann}. Penalized maximum likelihood estimation methods do not require the strong faithfulness assumption, which is known to be restrictive in high dimensions \citep{Uhler2012GeometryOT}. 
However, exactly solving the corresponding problem suffers from high computational complexity as the number of directed acyclic graphs is super-exponential in the number of variables. For example, learning an optimal graph using dynamic programming takes about 10 hours for a medium-size problem with 29 nodes \citep{silander2012simple}. 

One can resort to greedy search methods in a score-based setting for faster computation with large graphs. A popular example is the greedy equivalence search algorithm \citep{chickering2002optimal}, which performs a greedy search on the space of completed partially directed acyclic graphs rather than directed acyclic graphs. \citet{chickering2002optimal} established the asymptotic consistency of the resulting estimate. Furthermore, under a strong assumption that the graph has a fixed degree, greedy equivalence search and its variants have polynomial-time complexity in the number of variables \citep{Chickering2020StatisticallyEG}. Despite their favorable properties, these algorithms do not generally recover the optimal scoring graph for any finite sample size. Consequently, in moderate sample size regimes, the sub-optimal model they learn can differ from the true model. Another popular greedy approach is Greedy Sparse Permutations \citep{Solus21permutation}, which greedily searches over the space of permutations to identify a maximally scoring ordering of the variables. Similar to greedy equivalence search, this technique can get stuck in a local optima.

Several computationally efficient approaches rely on finding a topological ordering of the random variables under homoscedastic noise, also known as the equal-variance condition. When the nodes exhibit a natural ordering, the problem of estimating directed graphs reduces to the problem of estimating the network structure, which can be efficiently solved using lasso-like methods \citep{Ali10}. Assuming equal variances, \citet{Chen19} proposed a top-down approach to obtain an ordering among conditional variances. Similarly, \citet{ghoshal18a} introduced a bottom-up method that selects sinks by identifying the smallest diagonal element in conditional precision matrices obtained using the estimator of \citet{cai2011constrained}, resulting in a topological ordering. These approaches are efficient and thus readily extendable to high-dimensional problems. They also provide better estimates than greedy methods, such as the greedy equivalence search. Consequently, they are currently two of the state-of-the-art approaches for causal structure learning with homoscedastic noise. While being computationally efficient, these methods heavily rely on the equal-variance assumption. With non-equal variances, the topological ordering of variances no longer exists, and the performance of these methods deteriorates \citep{kucukyavuz2022consistent}. 

Integer and mixed integer programming formulations provide computationally efficient and rigorous optimization techniques for learning directed acyclic graphs. Typically, these methods consider a penalized maximum likelihood estimator over the space of directed acyclic graphs and cast the acyclicity condition as a constraint over a collection of binary variables. Such a formulation enables the use of branch-and-bound algorithms for fast and accurate learning of graphs, both with discrete \citep{Bartlett17discrete, cussens2017polyhedral, cussens2017bayesian} and continuous \citep{kucukyavuz2022consistent} data. More specifically, in the context of continuous Gaussian data, \cite{kucukyavuz2022consistent} impose homoscedastic noise assumption to develop their procedure. The homoscedastic noise assumption results in both computational and statistical simplifications. From the computational perspective, the penalized maximum likelihood estimator can be expressed as a mixed integer programming formulation, with a convex quadratic loss function and a regularization penalty subject to linear constraints. From the statistical perspective, the homoscedastic noise assumption implies full identifiability of the underlying directed acyclic graph \citep{Peters_2013} and it helps avoid the typical challenge of non-identifiability or identifiability of what is known as the completed partially directed acyclic graph from observational data. To the best of our knowledge, no integer programming formulation exists for continuous data without the stringent assumption of homoscedastic noise.


Mixed-integer programming has also been used for learning undirected graphical models, including Gaussian graphical models  \citep{Bertsimas2019certifiably, behdin2023sparse, misra2020information}. However, undirected graphical models need not satisfy the highly non-convex acyclicity constraint. Thus, learning directed acyclic graphs is often much harder than learning undirected graphical models and requires substantially different mixed-integer programming formulations.

\subsection{Our Contributions}
\label{sec:our_contributions}
We propose a mixed-integer programming formulation for learning directed acyclic graphs from general Gaussian structural equation models. Our estimator, which consists of a negative log-likelihood function that is a sum of a convex logarithmic term and a convex quadratic term, allows for heteroscedastic noise and reduces to the approach of \cite{kucukyavuz2022consistent} if the noise variances are constrained to be identical. Using a layered network formulation \citep{Manzour21}, as well as tailored cutting plane methods and branch-and-bound techniques, we provide a computationally efficient approach to obtain a solution that is guaranteed to be accurate up to a pre-specified optimality gap. By connecting the optimality gap of our mixed-integer program to the statistical properties of the estimator, we establish an early stopping criterion under which we can terminate the branch-and-bound procedure and attain a solution that recovers the Markov equivalence class of the underlying directed acyclic graph. Compared to state-of-the-art benchmarks in non-identifiable instances, we demonstrate empirically that our method generally exhibits superior estimation performance and displays greater robustness to variations in the level of non-identifiability compared to methods relying on the identifiability assumption. 
The improved performance is highlighted in the synthetic example of  Fig.~\ref{fig:alpha_results}, where our method is compared with the high-dimensional bottom-up approach of \cite{ghoshal18a}, the greedy sparsest permutations method of \cite{Solus21permutation}, the top-down approach of \cite{Chen19}, the mixed-integer second-order conic program of \cite{kucukyavuz2022consistent}, and the greedy equivalence search method of \cite{chickering2002optimal}.
For methods that support the use of superstructures as input---including \cite{kucukyavuz2022consistent}, \cite{Solus21permutation}, and \citet{chickering2002optimal}---we provide the estimated moral graphs as input. 
Here, we randomly select the noise variances from the interval $[4-\rho, 4+\rho]$, 
where larger $\rho$ leads to greater heteroscedasticity. Each method is implemented with a sample size of $n=100$ over 30 independent trials. The results show that in contrast to our method, the performance of competing approaches deteriorates under strong violations of the homoscedasticity assumption; see \S{\ref{subsec:var}} for more details. 


\begin{figure}[t]
    \centering
    \subfloat[$m=10$]{
        \includegraphics[scale=0.33]{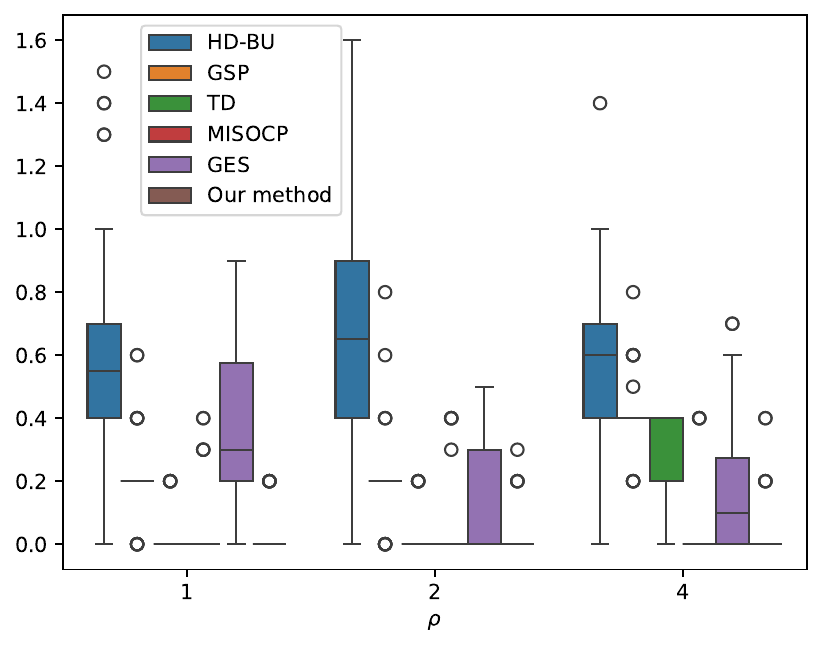}}
    \subfloat[$m=15$]{
        \includegraphics[scale=0.33]{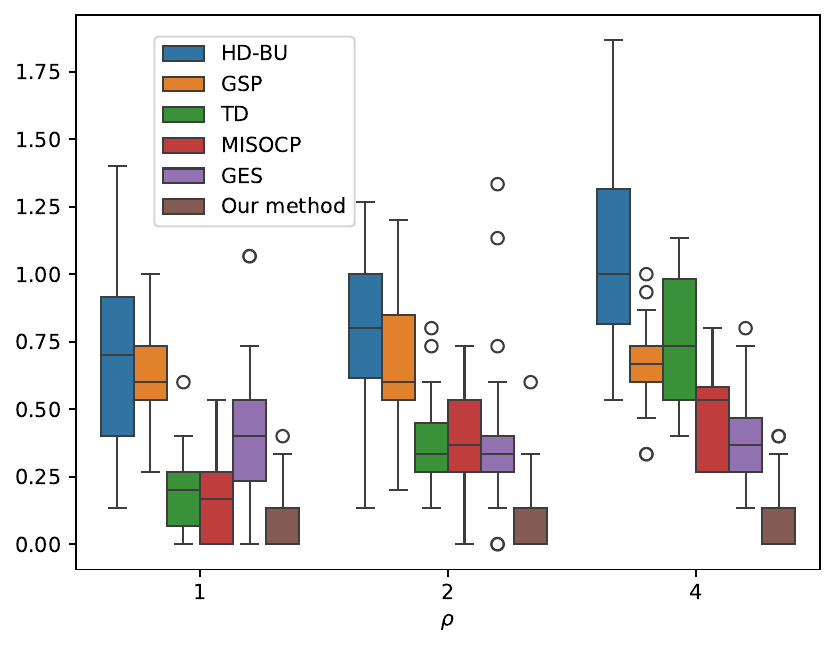}}
    \subfloat[$m = 20$]{
        \includegraphics[scale=0.33]{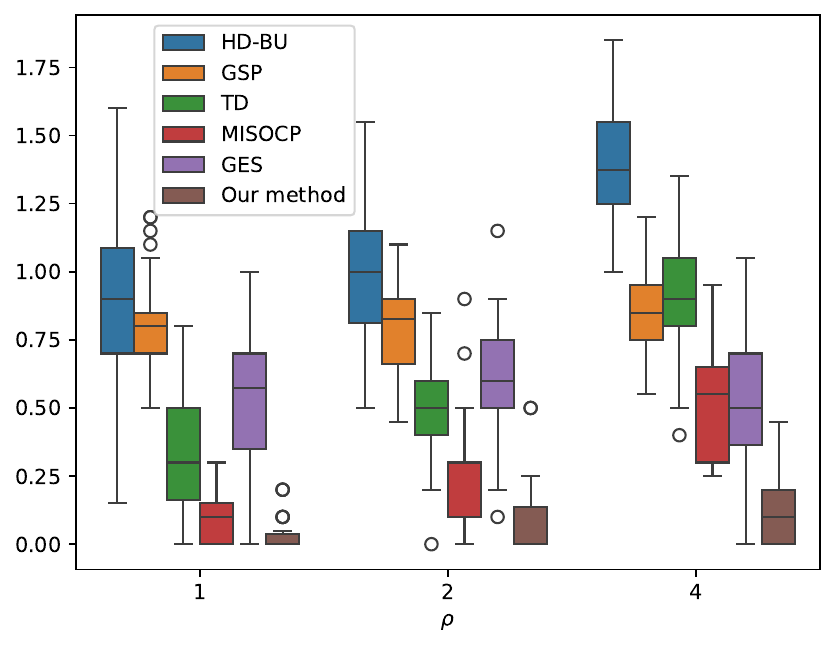}}
    \caption{{Box plots of scaled $d_{\mathrm{cpdag}}$ for our methods and benchmarks with $\rho = 1,2, 4$ and number of nodes $m=10, 15, 20$,  respectively. Here, $d_{\mathrm{cpdag}}$ values are scaled by the total number of edges in the true underlying directed acyclic graph. \hdbu, high-dimensional bottom-up; \gsp, greedy sparsest permutations; \td, top-down;   \misocp, mixed-integer second-order conic program; \ges, greedy equivalence search; $d_{\mathrm{cpdag}}$, differences between true and estimated completed partially directed acyclic graphs with smaller values being better; see \S\ref{sec:experiments}}.  
 } 
    \label{fig:alpha_results}
\end{figure}

\subsection{Notations and Definitions}
A directed acyclic graph $\mathcal{G} = (V,E)$ among $m$ nodes consists of vertices $V = \{1,2,\dots,m\}$ and directed edge set $E \subseteq V \times V$, where there are no directed cycles. We denote $\mathrm{MEC}(\mathcal{G})$ to be the Markov equivalence class of $\mathcal{G}$, consisting of directed acyclic graphs that have the same skeleton and same v-structures as $\mathcal{G}$. The skeleton of $\mathcal{G}$ is the undirected graph obtained from $\mathcal{G}$ by substituting directed edges with undirected ones. Furthermore, nodes $i,j$, and $k$ form a v-structure if $(i,k) \in E$ and $(j,k) \in E$, and there is no edge between $i,j$. The Markov equivalence class can be compactly represented by a completed partially directed acyclic graph, which is a graph consisting of both directed and undirected edges.  A completed partially directed acyclic graph has a directed edge from a node $i$ to a node $j$ if and only if this directed edge is present in every directed acyclic graph in the associated Markov equivalence class. A completed partially directed acyclic graph has an undirected edge between nodes $i$ and $j$ if the corresponding Markov equivalence class contains directed acyclic graphs with both directed edges from $i$ to $j$ and from $j$ to $i$. 
For a matrix $B \in \mathbb{R}^{m \times m}$, we denote $\mathcal{G}(B)$ to be the directed graph on $m$ nodes such that the directed edge from $i$ to $j$ appears in $\mathcal{G}(B)$ if and only if $B_{ij} \neq 0$. We denote the identity matrix by $I$ with the size being clear from context. The collection of positive-definite diagonal matrices is denoted by $\mathbb{D}_{++}^m$. For a matrix $M \in \mathbb{R}^{m \times m}$, $\|M\|_F := \{\sum_{i,j} M_{ij}^2\}^{1/2}$ denotes the Frobenius norm, $\|M\|_\infty = \max_{i,j}|M_{i,j}|$ denotes the max norm, and $\vertiii{M} := \max_{i}\sum_{j}|M_{ij}|$ denotes the $\ell_\infty$ opeator norm. We use the Bachman-Landau symbols $\mathcal{O}$ to describe the limiting behavior of a function. Furthermore, we denote $z \asymp 1$ to express $z = \mathcal{O}(1)$ and $1/z = \mathcal{O}(1)$.



\section{Problem setup}
\label{sec:setup}

\subsection{Modeling {F}ramework}
\label{sec:modeling_framework}
Consider an unknown directed acyclic graph whose $m$ nodes correspond to observed random variables $X \in \mathbb{R}^m$. We denote the directed acyclic graph by $\mathcal{G}^\star = (V,E^\star)$ where $V=\left\{1, \ldots, m\right\}$ is the vertex set and $E^\star \subseteq V \times V$ is the directed edge set. 

We assume the random variables $X$ satisfy a linear structural equation model:
\begin{equation*}\label{eq:compact_sem_model}
    X={B^\star}^\T{X}+ \epsilon.
\end{equation*}
Here, the connectivity matrix $B^\star \in \mathbb{R}^{m \times m}$ is a matrix with zeros on the diagonal and $B^\star_{jk} \neq 0$ if $(j,k) \in E^\star$. Further, $\epsilon$ is a mean-zero random vector with independent coordinates where the variances of individual coordinates can, in general, be different. {Without loss of generality, we assume that all random variables are centered}. Thus, each variable $X_j$ in this model can be expressed as the linear combination of its parents, the set of nodes with directed edges pointing to $j$, plus independent centered noise. Our objective is to estimate the matrix $B^\star$, or an equivalence class if the underlying model is not identifiable, as its sparsity pattern encodes the structure of the graph $\mathcal{G}^\star$.  To arrive at our procedure, we model the random variable $\epsilon \sim \mathcal{N}(0,\Omega^\star)$ to be Gaussian where noise variance matrices $\Omega^\star$ are positive definite and diagonal. Modeling the data to be Gaussian has multiple important implications. First, the structural equation model \eqref{eq:compact_sem_model} is parameterized by the connectivity matrix $B^\star$ and the noise variances $\Omega^\star$. Second, the random variables $X$ are distributed according to $\mathcal{P}^\star = \mathcal{N}(0,\Sigma^\star)$ where $\Sigma^\star = (I-B^\star)^\invT\Omega^\star(I-B^\star)^\inv$. Throughout, we assume that the distribution $\mathcal{P}^\star$ is non-degenerate, or equivalently, $\Sigma^\star$ is positive definite.   

In the effort to estimate the parameters $(B^\star,\Omega^\star)$, one faces a challenge with identifiability: there may be multiple structural equation models that are compatible with $\mathcal{P}^\star$. Specifically, consider any connectivity matrix $B$ such that the associated graph $\mathcal{G}(B)$ is a directed acyclic graph, and noise variance matrix $\Omega \in \mathbb{D}_{++}^m$. Then, the structural equation model entailed by $(B,\Omega)$, yields an equally representative model as the one entailed by the population parameters $(B^\star,\Omega^\star)$. Thus, the following question naturally arises: what is the set of equivalent directed acyclic graphs? The answer is that under an assumption called \emph{faithfulness}, the sparsest directed acyclic graphs that are compatible with $\mathcal P^\star$ are precisely $\mathrm{MEC}(\mathcal{G}^\star)$, the Markov equivalence class of $\mathcal{G}^\star$ {\citep{verma1990equivalence}}. 
Since parsimonious models are generally more desirable, our objective is to estimate $\mathrm{MEC}(\mathcal{G}^\star)$ from data. In the next sections, we describe a maximum likelihood estimator, and subsequently an equivalent mixed-integer programming framework, and present its optimality and statistical guarantees.

\subsection{An {I}ntractable Maximum-Likelihood Estimator}
We assume we have $n$ samples of $X$ which are independent and identically distributed. 
We denote $\hat{\Sigma}_n$ to be the sample covariance matrix of the data. Consider a Gaussian structural equation model parameterized by connectivity matrix $B$ and noise variance $\Omega$ with $D = \Omega^{-1}$. The parameters $(B,D)$ specify the following precision, or inverse covariance, matrix:
$$\Theta = \Theta(B,D) = (I-B){D}(I-B)^\T.$$
The negative log-likelihood of this structural equation model is proportional to
$$\ell_n(\Theta) = \mathrm{trace}(\Theta\hat{\Sigma}_n)-\log\det(\Theta).$$
Naturally, we seek a model that not only has a small negative log-likelihood but is also specified by a sparse connectivity matrix, containing few nonzero elements. Thus, we deploy the following $\ell_0$-penalized maximum likelihood estimator with regularization parameter $\lambda \geq 0$:
\begin{equation}
    \min\limits_{B \in \mathcal{B}, D \in \mathbb{D}^m_{++}} \ell_n\{(I-B) D (I-B)^\T\} + \lambda^2\norm{B}_{\ell_0}, 
\label{Problem:original}
\end{equation}
where $\mathcal{B} = \{B\in\R^{m \times m} \mid \mathcal{G}(B) \text{ is a directed acyclic graph}\}$ {is a complicated (non-convex) set} and $\|B\|_{\ell_0}$ denotes the number of non-zeros in the connectivity matrix $B$. 
A few remarks are in order. First, the estimator \eqref{Problem:original} is equivalent to one proposed and analyzed in \cite{vgBuhlmann}. Second, when the diagonal term $D$ is restricted to be a multiple of identity, the model will have homoscedastic noise. As desired, the resulting estimator reduces to the one considered in \cite{kucukyavuz2022consistent}. Finally, $\ell_0$ regularization is generally preferred over $\ell_1$ regularization in the objective of \eqref{Problem:original}. In particular, unlike the $\ell_1$ penalty, $\ell_0$ regularization preserves the important property that equivalent directed acyclic graphs---those in the same Markov equivalence class---have the same penalized likelihood score  \citep{vgBuhlmann}.

\cite{vgBuhlmann} demonstrated that the solution of \eqref{Problem:original} has desirable statistical properties; however, solving it is in general intractable. As stated, the likelihood function $\ell_n(\cdot)$ is a non-linear and non-convex function of the parameters $(B,D)$, and is thus not amenable to standard mixed-integer programming optimization techniques. In the following section, after a change of variables, we obtain a reformulation with a convex negative log-likelihood that enables a mixed-integer programming framework for our problem.

 \section{Our convex mixed-integer programming framework}
\label{sec:formulation}
\subsection{An Equivalent Reformulation}
Consider the following optimization problem:
\begin{equation}
\min\limits_{\Gamma \in \mathbb{R}^{m \times m}: \Gamma-\mathrm{diag}(\Gamma) \in \mathcal{B}} \sum_{i=1}^m-2\log(\Gamma_{ii})+\mathrm{tr}(\Gamma\Gamma^\T\hat{\Sigma}_n) + \lambda^2\norm{\Gamma-\mathrm{diag}(\Gamma)}_{\ell_0},
\label{Problem:3}
\end{equation}
where $\mathrm{diag}(\Gamma)$ is a diagonal $m \times m$ matrix consisting of the diagonal entries of $\Gamma$. The following proposition establishes the equivalence between \eqref{Problem:original} and \eqref{Problem:3}.
\begin{proposition}
Formulations \eqref{Problem:original} and \eqref{Problem:3} have the same minimum objective value. Further:
\begin{enumerate}
\item[(a)] for any minimizer $(\hat{B}^\mathrm{opt},\hat{D}^\mathrm{opt})$ of \eqref{Problem:original}, the parameter $\hat{\Gamma}^\mathrm{opt} = (I-\hat{B}^\mathrm{opt})(\hat{D}^\mathrm{opt})^{1/2}$ is a minimizer of \eqref{Problem:3};
\item[(b)] for any minimizer $\hat{\Gamma}^\mathrm{opt}$ of \eqref{Problem:3}, let $\hat{D}^\mathrm{opt} = \mathrm{diag}(\hat{\Gamma}^\mathrm{opt})^2$ and $\hat{B}^\mathrm{opt} = I-\hat{\Gamma}^\mathrm{opt}(\hat{D}^\mathrm{opt})^{-1/2}$. Then, the parameter set $(\hat{B}^\mathrm{opt},\hat{D}^\mathrm{opt})$ is a minimizer of \eqref{Problem:original}.
\end{enumerate}
\label{prop:main}
\end{proposition}
The proof of Proposition~\ref{prop:main} is in Appendix~\ref{proof:prop_main}. This result states that instead of \eqref{Problem:original}, one can equivalently solve the reformulated optimization problem \eqref{Problem:3}. Furthermore, due to the equivalence of the minimizers, the nonzero sparsity pattern in the off-diagonal of an optimal solution $\hat{\Gamma}^\mathrm{opt}$ of \eqref{Problem:3} encodes the same directed acyclic graph structure as a minimizer $\hat{B}^\mathrm{opt}$ of \eqref{Problem:original}. The formulation \eqref{Problem:3} replaces the non-convex negative log-likelihood term of \eqref{Problem:original} with the term $\sum_{i=1}^m-2\log(\Gamma_{ii})+\mathrm{tr}(\Gamma\Gamma^\T\hat{\Sigma}_n)$, which is a convex function of $\Gamma$ consisting of a quadratic term plus sum of univariate logarithm terms. These appealing properties of the estimator \eqref{Problem:3} 
enable an efficient mixed-integer programming framework, which we discuss next.

\subsection{A Convex Mixed-Integer Program}
\label{sec:mixed_integer}
Various integer programming formulations have been proposed to encode acyclicity constraints, including the cutting plane method \citep{Grtschel1985OnTA,wolsey1999integer}, topological ordering \citep{Park17}, and layered network formulation \citep{Manzour21}. Here, we adopt the layered network formulation as it has been shown to perform the best with continuous data; see \cite{Manzour21} for more details. 

Before we present our formulation, we first introduce some notation. We denote the edge set $E_\text{super}$ as a \emph{superstructure} of $E^\star$ (true edge set) with the property that $E^\star \subseteq E_\text{super}$. Here, the superstructure $E_\text{super}$ may be bi-directional and consist of cycles with no self-loops. When ${E}_\text{super}$ is sparse, we can significantly reduce the search space. A good candidate for a superstructure is the moral graph corresponding to the true directed acyclic graph, which can be readily estimated from data using the \emph{graphical lasso} {\citep{friedman2007sparse}}.

Following previous work on the layered network formulation,  to efficiently encode the acyclicity constraint, we add two sets of decision variables to the optimization model. The first is the set of binary variables $\{g_{jk} \in \{0,1\}:(j,k) \in E_\text{super}\}$. These variables are used to represent the presence or absence of edges in the estimated directed acyclic graph. The second is the set of continuous variables $\{\psi_{k} \in [1,m]: k \in \{1,\ldots,m\}\}$ representing the \emph{layer value} of each node in the estimated graph with the property that an ancestor of a node in the graph has a higher layer value. Formally, the layered network formulation of Problem \eqref{Problem:3} is:
\vspace{-0.3cm} 
\begin{subequations}
\label{problem:big-M}
\begin{align}
\label{eq:big-M_obj}\min\limits_{\substack{\Gamma \in \mathbb{R}^{m \times m},\psi \in [1,m]^m\\ g \in \{0,1\}^{|E_\text{super}|}}} & \sum_{i=1}^m-2\log(\Gamma_{ii})+\mathrm{tr}(\Gamma\Gamma^\T\hat{\Sigma}_n) + \lambda^2\sum_{(j,k) \in {E}_\text{super}}g_{jk} \\
\text { s.t. }
& -M g_{j k} \leq \Gamma_{j k} \leq M g_{j k}  \quad ((j, k) \in {E}_\text{super}), \label{eq:big_M}\\
&M \geq \Gamma_{ii} \quad (i=1,\ldots, m), \label{eq:big_M-diag}\\
\label{eq:acyclic}&1 - m + m g_{jk} \leq \psi_k - \psi_j \quad ((j, k) \in{E}_\text{super}).
\end{align} 
\end{subequations}
\vspace{-0.5cm}\\
The constraints \eqref{eq:big_M} are so-called \emph{big-M constraints} \citep{Park17,Manzour21} that bound the magnitude of the entries of $\Gamma$ by a large $M$ to ensure that when $g_{jk}=0$, $\Gamma_{jk}=0$ and when $g_{jk}=1$ this constraint is redundant; we will explore choices for $M$ below. While constraints \eqref{eq:big_M-diag} are redundant for diagonal entries for large enough $M$, they are beneficial for computational efficiency. Furthermore, since more degrees of freedom yield better likelihood fit, $g_{jk}=1$ generally yields $\Gamma_{jk} \neq 0$, and thus the regularization term $\|\Gamma - \mathrm{diag}(\Gamma)\|_{\ell_0}$ in \eqref{Problem:3} is identical to $\sum_{(j,k) \in {E}_\text{super}}g_{jk}$ in \eqref{problem:big-M}. Finally, the constraint \eqref{eq:acyclic}, together with \eqref{eq:big_M}, ensures that $\Gamma$ encodes a directed acyclic graph. To see why, suppose there is a directed path in $\Gamma$ from node $j$ to node $k$ and a directed path from $k$ to $j$. The constraint \eqref{eq:acyclic} then ensures that $\psi_{k} \geq \psi_{j}+1$ and $\psi_{j} \geq \psi_{k}+1$, resulting in a contradiction. 
We have thus shown that for $M$ large enough, e.g., $M$ equaling the maximum nonzero entry, in magnitude, of the optimal $\Gamma$ in \eqref{Problem:3}, the mixed-integer program \eqref{problem:big-M} is equivalent to \eqref{Problem:3}. One can use heuristic approaches \citep{Park17,kucukyavuz2022consistent} to select the value of $M$. Specifically, we solve the problem without any cycle prevention constraints and obtain solution $\hat{\Gamma}$, then let $M=2\max_{(i,j) \in {E}_\text{super}} |\hat{\Gamma}_{ij}|$.

\subsection{Perspective Strengthening for a Tighter Formulation}
Using a concept known as perspective strengthening \citep{cui13stronger,kucukyavuz2022consistent,wei2020convexification,wei2021ideal,Wei2023}, we can tighten the constraint set of \eqref{problem:big-M} without changing the optimal objective value. Such a tighter formulation can speed up computations by providing a better lower bound in the branch-and-bound procedure. 

Specifically, let $\delta\in \R^m_+$ be a non-negative vector with the property that $\hat{\Sigma}_n - D_\delta \succeq 0$, where $D_\delta = \mathrm{diag}(\delta_1, ..., \delta_m)$. By splitting the quadratic term $\Gamma\Gamma^\T \hat{\Sigma}_n = \Gamma\Gamma^\T (\hat{\Sigma}_n - D_\delta) + \Gamma\Gamma^\T D_\delta$ and the fact that $\mathrm{tr}(\Gamma\Gamma^\T D_\delta) = \sum_{j=1}^m\sum_{k=1}^m \delta_j\Gamma_{jk}^2$, we can express \eqref{eq:big-M_obj} as $-2\sum_{i=1}^m\log(\Gamma_{ii})+\mathrm{tr}(\Gamma\Gamma^\T Q) + \mathrm{tr}(\Gamma\Gamma^\T{D}_\delta) + \lambda^2\sum_{(j,k) \in {E}_\text{super}}g_{jk}$,  where $Q =  \hat{\Sigma}_n - D_\delta$. We add a new set of non-negative continuous variables $s_{jk}$ to represent $\Gamma_{jk}^2$, resulting in the following formulation:
\begin{subequations}
\allowdisplaybreaks
\label{problem:micp}
\begin{align}
\label{problem_obj:a}\min\limits_{\substack{\Gamma \in \mathbb{R}^{m \times m},\psi \in [1,m]^m\\ g \in \{0,1\}^{|E_\text{super}|}\\s \in \mathbb{R}^{|E_\text{super}|}}}  & \sum_{i=1}^m-2\log(\Gamma_{ii})+\mathrm{tr}(\Gamma\Gamma^\T Q) + \sum\limits_{(j,k)\in {E}_\text{super}}\delta_j s_{jk} + \sum_{i=1}^m \delta_i s_{ii} +\lambda^2\sum_{(j,k) \in {E}_\text{super}}g_{jk} \\
\label{problem_const:b}\text { s.t. } &  -M g_{j k} \leq \Gamma_{j k} \leq M g_{j k}  \quad ((j, k) \in {E}_\text{super}), \\
& M \geq \Gamma_{ii} \quad (i=1,\ldots, m), \\
&1 - m + m g_{jk} \leq \psi_k - \psi_j  \quad ((j, k) \in {E}_\text{super})\label{eq:dag_new},\\
& \label{eqn:s_cont} s_{j k}g_{j k} \geq \Gamma_{jk}^2 \quad((j, k) \in {E}_\text{super}), \quad s_{ii} \geq \Gamma_{ii}^2 \quad(i = 1,\dots,m),\\
&s_{j k} \leq M^2 g_{j k} \quad ((j, k) \in {E}_\text{super}),\quad s_{ii} \leq M^2 \quad (i=1,\ldots, m).
\end{align} 
\end{subequations}
To establish that the optimal objective values of \eqref{problem:big-M} and \eqref{problem:micp} are identical, it suffices to show that the optimal set of variables $(\hat{s}^\mathrm{opt},\hat{\Gamma}^\mathrm{opt})$ in \eqref{problem:micp} satisfy $\hat{s}^\mathrm{opt}_{jk} = (\hat{\Gamma}^\mathrm{opt}_{jk})^2$ for $(j,k) \in E_\text{super}$ and $\hat{s}^\mathrm{opt}_{ii} = (\hat{\Gamma}^\mathrm{opt}_{ii})^2$ ($i = 1,\dots,m$). From the constraints \eqref{problem_const:b} and \eqref{eqn:s_cont}, it follows that $s_{jk} \geq \Gamma_{jk}^2$ for $(j,k) \in E_\text{super}$. Since $\delta_j$ is non-negative in the objective \eqref{problem_obj:a}, we can get $\hat{s}^\mathrm{opt}_{jk} = (\hat{\Gamma}^\mathrm{opt}_{jk})^2$ for $(j,k) \in E_\text{super}$ and $\hat{s}^\mathrm{opt}_{ii} = (\hat{\Gamma}^\mathrm{opt}_{ii})^2$ ($i = 1,\dots,m$). The constraints $s_{j k} \leq M^2 g_{j k}$ and $s_{ii} \leq M^2$ are simply added to improve the computational efficiency. 

The continuous relaxation of constraint set of \eqref{problem:micp}, replacing the integer constraints with $g \in [0,1]^{|E_\text{super}|}$, is contained in the continuous relaxation of the one from \eqref{problem:big-M}: for every set of feasible variables in \eqref{problem:micp}, there exists a set of feasible variables in \eqref{problem:big-M} yielding the same objective value. This fact, which follows from the analysis of \cite{cui13stronger}, leads to better lower bounds in the branch-and-bound process. Hence, throughout, we use the formulation \eqref{problem:micp}. 

In the above formulation, there is some flexibility in the choice of the vector $\delta$. As larger values of $\delta$ lead to smaller continuous relaxation of the constraint set of \eqref{problem:micp} \citep{FRANGIONI2007181}, we choose $\delta$ by maximizing $\sum_{i=1}^m \delta_i$ subject to $\hat{\Sigma}_n - \mathrm{diag}(\delta) \succeq 0$ for $\delta_i \geq 0, i=1,\ldots, m$. This is a convex semi-definite program that can be solved efficiently. 

\section{Practical aspects: outer approximation and early stopping}

\subsection{Outer Approximation to Handle the Logarithm}
\label{sec:solution method}
Formulation \eqref{problem:micp} is a convex mixed-integer model that can be solved using the current optimization solvers, such as \texttt{Gurobi}. However, \texttt{Gurobi} handles the log terms in the objective by adding a new general constraint that uses a piecewise-linear approximation of the logarithmic function at predetermined breakpoints. Unless we use a large number of pieces in the piecewise-linear functions, which is computationally intensive, this approximation often leads to solutions that violate some constraints. To address this challenge, in this section, we describe an outer approximation method, based on ideas first proposed by \citet{Duran1986AnOA} {and widely applied across various problems \citep{Bertsimas2019certifiably, behdin2021sparse}}. 

Outer approximation is a cutting plane method that finds the optimal solution by constructing a sequence of piece-wise affine lower bounds for the logarithmic term in the objective function. We first replace $-2\log(\Gamma_{ii})$ with continuous variables $T_i \in \R, i=1,\ldots, m$ so that the objective function \eqref{problem_obj:a} becomes a quadratic function: 

\begin{equation*}
    c(T, \Gamma)= \sum_{i=1}^m T_i+\mathrm{tr}(\Gamma\Gamma^\T Q) + \sum\limits_{(j,k)\in {E}_\text{super}}\delta_j s_{jk} + \sum_{i=1}^m \delta_i s_{ii} +\lambda^2\sum_{(j,k) \in {E}_\text{super}}g_{jk}.
\end{equation*}
Without any constraint on $T$, the problem would be unbounded. Therefore, we iteratively add linear inequalities that are under-estimators of the logarithmic term to ensure that $T_{i} = -2\log(\Gamma_{ii})$ ($i = 1,\dots,m$) at the optimal solution. {We implement these cutting planes from outer approximation as lazy constraints and utilize \texttt{Gurobi} as our solver.}

The linear approximation of function $f(x) = -2\log(x)$ at a point $x_0$ is $f(x_0) + \nabla f(x_0) (x-x_0) = -2\log(x_0) - 2 (x - x_0) / x_0$. At iteration $t+1$, given the solution $\Gamma^{(t)}$, we add linear constraints $T_i \geq -2\log(\Gamma^{(t)}_{ii}) - 2(\Gamma_{ii}-\Gamma^{(t)}_{ii})/\Gamma^{(t)}_{ii}$ ($i = 1,\dots,m$). As $f(x)$ is convex, these linear approximations are under-estimators and they cut off the current solution unless $T^{(t)}_{i}$ equals $-2\log(\Gamma^{(t)}_{ii})$ ($i = 1,\dots,m$). The overall procedure is presented in Algorithm~\ref{alg:oa} and a simple illustrative example is given in Appendix~\ref{oa_example}.

\begin{algorithm}
\caption{Pseudocode for outer approximation}
\begin{tabbing}
\qquad Input: Sample covariance $\hat{\Sigma}_n$ and regularization $\lambda \in \R_+$\\
\qquad Initialize: $t\gets 1$; 
$\Gamma^{(t)}\gets \text{starting point}$;
$T^{(t)}\gets - \infty$\\
\qquad While $T^{(t)}_j < -2\log(\Gamma^{(t)}_{jj})$ for some $j = 1,\ldots,m$ do\\
\qquad \qquad  $\Gamma^{(t+1)}, T^{(t+1)} \gets 
  \arg\min\limits_{T, \Gamma, g, \psi} c(T, \Gamma) \text{ s.t. } T_j \geq -2\log(\Gamma^{(i)}_{jj}) - \frac{2}{\Gamma^{(i)}_{jj}}(\Gamma_{jj}-\Gamma^{(i)}_{jj})$,\\\hspace{2.5in}$(i= 1,\ldots,t; j=1,\ldots,m)$, \eqref{problem_const:b}-\eqref{eq:dag_new}\\
\qquad \qquad  $t \gets t + 1$\\
\qquad $\hat{\Gamma}^\mathrm{opt} \gets \Gamma^{(t)}$\\
\qquad Output: $\hat{\Gamma}^\mathrm{opt} \in \R^{m\times m}$ \\
\label{alg:oa}
\end{tabbing}
\label{alg:oa}
\end{algorithm}

\subsection{Early Stopping}
\label{sec:early_stopping}
During the branch-and-bound algorithm, we maintain a lower bound and an upper bound on the objective value of the objective function \eqref{problem_obj:a}.  Specifically, by relaxing the integer constraints to $g \in [0,1]^{|E_\text{super}|
}$, one can solve the convex problem easily and the optimal objective value of the relaxed problem provides a lower bound. If any integer variable is fractional in this solution, we can split the problem into two sub-problems by considering smaller or larger integer values for that variable. This process, which is known as the \emph{branch-and-bound} technique, creates a tree structure, where each node represents a problem and is connected to its potential sub-problems. If the relaxed solution of a node contains only integer values for all integer variables, then we have a feasible solution to the original problem, giving an upper bound of the objective, and the branching process terminates for that node. Throughout the optimization algorithm, we continuously update the upper and lower bounds. The optimality gap, $\gap$, of a solution is the difference between the upper and lower bounds of the objective at that solution and it should be $0$ at an optimal solution. Alternatively, we can stop the algorithm early, before reaching the optimal solution, when the optimality gap reaches a specified threshold. Due to the intrinsic computational complexity of solving mixed-integer programming problems, early stopping is often used in practice but without statistical justification, because the solution may be suboptimal.

\section{Theoretical results for early stopping}
\label{sec:theory}
In Section~\ref{sec:early_stopping}, we described how we can terminate the branch-and-bound algorithm to guarantee a desired optimality gap. In this section, we connect this optimality gap to the statistical properties of the terminated solution. This connection enables us to propose an optimality gap under which we can terminate the branch-and-bound procedure and attain a solution that is close to a member of true Markov equivalence class $\mathrm{MEC}(\mathcal{G}^\star)$ and is asymptotically consistent. Further, after suitable thresholding of the terminated solution, we show that the associated Markov equivalence class precisely coincides with the population Markov equivalence class. We consider the estimator \eqref{Problem:3} with the additional constraint that the edges are constrained to be within an input superstructure $E_\mathrm{super}$, where we assume $E^\star \subseteq \text{moral}(E^\star) \subseteq E_\mathrm{super}$ where $E^\star$ denotes the true edge set and $\text{moral}(E^\star)$ is the edge set corresponding to the true moral graph (the moral graph is the undirected graph of $\mathcal{G}^\star$ by adding edge between nodes that have common children and undirects all edges); Appendix\ref{sec:super_structure} shows how graphical lasso can yield superstructures that satisfy this property with high probability. We simply assume access to such a superstructure in the remainder. Furthermore, following the analysis
of \cite{vgBuhlmann}, we impose an additional constraint that $\|\Gamma_{,i}\|_{\ell_0} \leq \tilde{\alpha} n/\log m$ for all $i = 1,2,\dots,m$ and some constant $\tilde{\alpha}$.



Recall from \S{\ref{sec:modeling_framework}} that there may be multiple structural equation models that are compatible with the distributions $\mathcal{P}^{\star}$. Each equivalent structural equation model is specified by a directed acyclic graph; this directed acyclic graph defines a total ordering among the variables. Associated to each ordering $\pi$ is a unique structural equation model that is compatible with the distribution $\mathcal{P}^{\star}$. We denote the set of parameters of this model as $(\tilde{B}^{\star}(\pi),\tilde{\Omega}^{\star}(\pi))$. For the tuple $(\tilde{B}^{\star}(\pi),\tilde{\Omega}^{\star}(\pi))$, we define $\tilde{\Gamma}^{\star}(\pi) = \{I - \tilde{B}^{\star}(\pi)\}\tilde{\Omega}^{\star}(\pi)^{-1/2}$. We let $\Pi = \{\pi: \mathrm{support}(\tilde{B}^\star(\pi)) \subseteq E_\text{super}\}$. We will use the notation $s^\star = \|B^\star\|_{\ell_0}$ and $\tilde{s}= \tilde{s}^{\star}(\pi) = \|\tilde{B}^{\star}(\pi)\|_{\ell_0}$.

 \begin{assumption}\label{condition:condition4}{(sparsity of every equivalent causal model)  There exists some constant $\tilde{\alpha}$ such that for every $\pi \in \Pi$, $\|\tilde{B}^{\star}_{:j}(\pi)\|_{\ell_0} \leq \tilde{\alpha} n/\log m$.}
  \end{assumption}

\begin{assumption}(beta-min condition)\label{condition:condition5}
There exist constants $0 \leq \eta_1 < 1$ and $0 < \eta_0^2 < 1 - \eta_1$, such that for every $\pi \in \Pi$, the matrix $\tilde{B}^{\star}(\pi)$ has at least $(1 - \eta_1)\|\tilde{B}^{\star}(\pi)\|_{\ell_0}$ coordinates $k \neq j$ with $|\tilde{B}^{\star}_{kj}(\pi)| > \mathcal{O}\left[(\log m/n)^{1/2}\{(m/{s^\star})^{1/2}\lor1\} / \eta_0\right]$. 
\end{assumption}

\begin{assumption} (bounded minimum and maximum eigenvalues) The eigenvalues of $\Sigma^\star$ are bounded below by $\underline{\kappa}$ and above by $\bar{\kappa}$, respectively. Moreover, $\max_j \|\Sigma^\star_{j,}\|_1 \leq \kappa_1$ for some constant $\kappa_1$. Finally, let $N = [\Sigma^\star \otimes \Sigma^\star]_{E_\text{super},E_\text{super}}$ be the restriction of
$\Sigma^\star \otimes \Sigma^\star$ to coordinates in $E_\mathrm{super}$. Then, $\vertiii{N} > \kappa_2$ for some nonzero constant $\kappa_2>0$.
\label{assump:sample_cond}
\end{assumption}

\begin{assumption}(sufficiently large noise variances) For every permutation $\pi \in \Pi$, $\mathcal{O}(1)\geq \min_{j} \{\tilde{\Omega}^{\star}(\pi)\}_{jj} \geq  \mathcal{O}\{(s^\star\log m /n)^{1/2}\}$.
\label{ass:noise_lower_bound}
\end{assumption}

\begin{assumption}(faithfulness)\label{condition:faithfulness}
Every conditional independence relationship entailed in the underlying distribution of the variables is encoded in the population directed acyclic graph $\mathcal{G}^\star$. 
\end{assumption}
Here, Assumptions~\ref{condition:condition4}-\ref{assump:sample_cond} are similar to those in \cite{vgBuhlmann}. Assumption~\ref{ass:noise_lower_bound} is used to characterize the behavior of the early stopped estimate and is thus new relative to \cite{vgBuhlmann}. Assumption~\ref{condition:faithfulness} on faithfulness is used to connect our estimates to the population parameters in \eqref{eq:compact_sem_model}; we describe a relaxation of this condition shortly. 
{
\begin{theorem} Suppose Assumptions~\ref{condition:condition4}--\ref{condition:faithfulness} are satisfied with constants $\tilde{\alpha}$ and $\eta_0$ sufficiently small and that $n\geq\mathcal{O}(m)$. Let $\alpha_0 = (4/m) \land 0.05$. Suppose we let the optimality gap criterion of our algorithm be $\gap = \mathcal{O}(\lambda^2{s}^\star)$ and let $\hat{\Gamma}^\mathrm{early}$ be the early stopped estimate. Then, for $\lambda^2 \asymp \log m/n$,  letting $\hat{\pi}^\mathrm{early}$ be the ordering associated with $\hat{\Gamma}^\mathrm{early}$, we have that, with probability greater than $1-2\alpha_0$, $\|\hat{\Gamma}^\mathrm{early}-\tilde{\Gamma}^{\star}(\hat{\pi}^\mathrm{early})\|_F^2 \leq \mathcal{O}(\lambda^2s^\star)$, and $\|\tilde{\Gamma}^{\star}(\hat{\pi}^\mathrm{early})\|_{\ell_0} \asymp s^\star$.
\label{thm:main}
\end{theorem}}
The proof of Theorem~\ref{thm:main} is in Appendix\ref{proof:thm_1}. The result guarantees that our early stopping optimization procedure accurately estimates certain reordering of the population model. As described in Appendix\ref{sec:comparison_early_stopping_theory}, the optimal solution of \eqref{Problem:3} (without early stopping) exhibits a similar convergence (to the population quantity) as the solution obtained from early stopping. Thus, our result highlights that early-stopping provides computational gain, without sacrificing statistical accuracy; see Section~\ref{subsec:early_stopping} for numerical support.

For accurately estimating the edges of the population model, we need a stronger version of the beta-min condition \cite{vgBuhlmann}, dubbed the strong beta-min condition. Let $d_{\max} = \max_{i} |\{j: (i,j)\in E_\mathrm{super} \text{ or } (j,i) \in E_\mathrm{super}\}|$, where for sparse true moral graphs, one can expect $d_{\max} = \mathcal{O}(1)$; see Appendix\ref{sec:super_structure}. 
\begin{assumption}(strong beta-min condition)\label{condition:strong_beta}
There exists constant $0 < \eta_0^2 < 1/s^\star$, such that for any $\pi \in \Pi$, the matrix $\tilde{B}^{\star}(\pi)$ has all of its nonzero coordinates $(k,j)$ satisfy $|\tilde{B}^{\star}_{kj}(\pi)| > \mathcal{O}\{(d_{\max}m\log m/n)^{1/2}/ \eta_0\}$. 
\end{assumption}

With Assumption~\ref{condition:strong_beta} we can guarantee that the estimated model is close to a member of the Markov equivalence class of the underlying directed acyclic graph. For a member of the population Markov equivalence class, let $(B_\mathrm{mec}^\star,\Omega_\mathrm{mec}^\star)$ be the associated connectivity matrix and noise matrix that specify an equivalent model as \eqref{eq:compact_sem_model}. Furthermore, define $\Gamma_\mathrm{mec}^\star = (I-B_\mathrm{mec}^\star){\Omega_\mathrm{mec}^\star}^{-1/2}$. 

{\begin{theorem} Suppose that $\lambda^2 \asymp d_{\max}m\log m/n$, and assumptions of Theorem~\ref{thm:main} as well as Assumption~\ref{condition:strong_beta} hold. Suppose $d_{\max} \leq \mathcal{O}\{(n/\log m)^{1/2}\}$ and that $n \geq \mathcal{O}(m)$. Suppose we let the optimality gap criterion of our algorithm be $\gap= c\lambda^2$ for some $0<c<1$. Then, with probability greater than $1-2\alpha_0$, there exists a member of the population Markov equivalence class with associated parameter $\Gamma^\star_\mathrm{mec}$ such that i) $\|\hat{\Gamma}^\text{early}-\Gamma^\star_\mathrm{mec}\|_F^2 \leq \mathcal{O}(\lambda^2{s}^\star)$ and ii) we recover the population Markov equivalence class; that is, $\mathrm{MEC}(\hat{\mathcal{G}}) = \mathrm{MEC}({\mathcal{G}}^\star)$ where $\hat{\mathcal{G}} = \mathcal{G}\{\hat{\Gamma}^\mathrm{early}-\mathrm{diag}(\hat{\Gamma}^{\mathrm{early}})\}$ is the estimated directed acyclic graph.
\label{corr:equiv}
\end{theorem}}
The proof of Theorem~\ref{corr:equiv} is given in Appendix\ref{proof:thm_2}. This theorem guarantees that by taking $\lambda^2$ large enough, early stopping consistently estimates a population parameter corresponding to the underlying Markov equivalence class, and recovers the population Markov equivalence class. 

We remark that without the faithfulness condition in Assumption~\ref{condition:faithfulness}, we can guarantee that the early stopping procedure recovers what is known as the \emph{minimal-edge I-MAP}. The minimal-edge I-MAP is the sparsest set of directed acyclic graphs that induce a structural equation model that is compatible with the true data distribution.
Under faithfulness, the minimal-edge I-MAP coincides with the population Markov equivalence class \citep{vgBuhlmann}.

\section{Experiments}
\label{sec:experiments}
\subsection{Setup}

In this section, we illustrate the utility of our method over competing methods on synthetic and real directed acyclic graphs. The state-of-the-art approaches include the high-dimensional top-down approach in \cite{Chen19}, the high-dimensional bottom-up approach in \cite{ghoshal18a}, the mixed-integer second-order cone program in \cite{kucukyavuz2022consistent}, the PC algorithm in \cite{causalitybase}, the greedy sparsest permutations method in \cite{Solus21permutation}, and the greedy equivalence search method in \cite{chickering2002optimal}.  We supply the estimated moral graph as the input superstructure for our and \cite{kucukyavuz2022consistent}'s methods (see \S{\ref{sec:mixed_integer}}). The moral graph is estimated from data using the {graphical lasso} \citep{friedman2007sparse}; see Appendix\ref{sec:super_structure} for more description.
In Appendix\ref{sec:compare_benchmarks}, we also present results where the true moral graphs are used as input. All experiments are performed with a MacBook Air with an Apple M2 chip, 8-core CPU, and 8 GB of RAM with \texttt{Gurobi} 10.0.1 Optimizer. Our method and the method by \citet{kucukyavuz2022consistent} are implemented using \texttt{Python}. \citet{Chen19}'s top-down methods and \citet{ghoshal18a}'s bottom-up method are implemented in \texttt{R}.

As stated earlier, due to heteroscedastic noise, the underlying directed acyclic graph is generally identifiable up to the Markov equivalence class, represented by a completed partially directed acyclic graph. Thus, to evaluate the performance of the methods, we use the metric $d_\text{cpdag}$, which is the number of different entries between the unweighted adjacency matrices of the two completed partially directed acyclic graphs. 

Unless otherwise specified, we set the desired optimality gap in our branch-and-bound algorithm to zero. We terminate the algorithm if our branch-and-bound algorithm does not achieve the desired optimality gap within $50m$ seconds. We report the solution time (in seconds) and achieved relative optimality gap, $\rgap = (\text{upper bound} - \text{lower bound})/\text{lower bound}$, where $\text{upper bound}$ is the objective value associated with the best feasible solution and $\text{lower bound}$ represents the best obtained lower bound during the branch-and-bound process.

We use the Bayesian information criterion to choose the parameter $\lambda$. In our context, the Bayesian information criterion score is
$-2n\sum_{i=1}^m\log \hat{\Gamma}_{ii}+n\mathrm{tr}(\hat{\Gamma}\hat{\Gamma}^\T\hat{\Sigma}_n)+ k\log n$,
where $k$ is the number of nonzero entries in the estimated parameter $\hat{\Gamma}$. From our theoretical guarantees in \S{\ref{sec:theory}}, $\lambda^2$ should be on the order ${\log m/n}$. Hence, we choose $\lambda$ with the smallest Bayesian information criterion score among $\lambda^2 = c^2 \log m/n$ $(c=1,\ldots,15)$.
 The code to reproduce all the experiments is available at \texttt{https://github.com/AtomXT/MICP-NID}.

\subsection{Comparison to Other Benchmarks}
\label{subsec:comparison_benchmarks}
We compare the performance of our method with the benchmark methods on twelve publicly available networks sourced from \cite{Manzour21} and the Bayesian Network Repository (bnlearn). These networks vary in size, ranging from $m = 6$ to $m = 70$ nodes, and true number of edges, $s^\star$ ranging from 6 to 123. We indicate the size of each network next to the name of the network, e.g., Hepar2.70.123 indicates that the Hepar2 network has 70 nodes and 123 edges. For a given network structure, we generate $n = 500$ independent and identically distributed samples from \eqref{eq:compact_sem_model} where the nonzero entries of $B^\star$ are drawn uniformly at random from the set $\{-0.8, -0.6, 0.6, 0.8\}$ and diagonal entries of $\Omega^\star$ are chosen uniformly at random from the set $\{0.5, 1, 1.5\}$. In \S{\ref{subsec:var}}, we will explore a larger range of noise variances. We also present experiments on non-Gaussian error in Appendix\ref{sec:non-Gaussian}.

Table~\ref{tab:compare_benchmarks_est} summarizes the performance of all methods, averaged over {$30$} independent trials. Here, the symbol $*$ indicates that the achieved optimality gap is zero, so an optimal solution is found. The symbol T in the time column means that the time limit of $50m$ seconds was reached. We report the structural hamming distances of the undirected skeleton of the true directed acyclic graph and the corresponding skeleton of the estimated network, the true positive rate, and the false positive rate in Appendix\ref{sec:Other_metrics}. Compared to the top-down approach of \cite{Chen19}, the bottom-up approach of \cite{ghoshal18a}, the PC algorithm of \cite{causalitybase} and the greedy sparsest permutation method of \cite{Solus21permutation}, our method produces more accurate estimates, giving smaller $d_\mathrm{cpdag}$, and provides optimality guarantees, $\rgap$. 
For comparison, error standard deviations are also reported.
We also tested the high-dimensional top-down approach of \cite{Chen19}. However, since it consistently underperforms the top-down approach in these instances, we do not include it in the table. 

We observe that when our algorithm terminates at proven optimality in 7 out of 12 instances, it also produces more accurate estimates than the method of \cite{kucukyavuz2022consistent}. The improved statistical performance is because our method accounts for heteroscedastic noise or non-identifiability, while the other methods do not. { Nonetheless, we note that our optimization method can be slower than the mixed-integer second-order conic method of \cite{kucukyavuz2022consistent}, as it contains logarithmic terms in its objective to handle heteroscedastic noise as opposed to a simpler convex quadratic objective for the homoscedastic case. As compared to greedy equivalence search \cite[]{chickering2002optimal}, the proposed method shows a lower mean $d_{\mathrm{cpdag}}$ for all but one of the graphs and smaller standard deviation for all of the graphs; furthermore, for the graph ``cloud" where the mean $d_{\mathrm{cpdag}}$ of greedy equivalence search is smaller than ours, the difference is very small. The improved performance is attributable to the fact that while both methods aim to maximize the $\ell_0$-penalized log-likelihood function, the heuristic approach of the greedy equivalence search may not find the optimal solution, resulting in larger errors. To further elucidate the comparison with the greedy equivalence search, we present box plots of $d_{\mathrm{cpdag}}$ values in Appendix~\ref{sec:Other_metrics}, and apply a Wilcoxon signed-rank test to show that for all but three of the graphs, the improvement of our approach is statistically significant. 
Since the PC and greedy sparsest permutation algorithms underperform compared to the greedy equivalence search, their results are reported in Appendix\ref{app:pc-gsp}.}

\begin{table}[ht]
\def~{\hphantom{0}}
\caption{{Comparison of mixed-integer convex program and other benchmarks using estimated superstructure, if applicable}}
    \centering
    \resizebox{\columnwidth}{!}{\begin{tabular}{lcccccccccccccc}
    
        & \multicolumn{2}{c}{\hdbu} & \multicolumn{2}{c}{\td} & \multicolumn{3}{c}{\misocp} & \multicolumn{2}{c}{\ges} & \multicolumn{3}{c}{\micp}\\
       Network.$m$.$s^\star$ & Time  & $d_{\mathrm{cpdag}}$ & Time  & $d_{\mathrm{cpdag}}$ & Time  & \rgap & $d_{\mathrm{cpdag}}$ & Time & $d_{\mathrm{cpdag}}$ & Time  & \rgap & $d_{\mathrm{cpdag}}$ \\
       Dsep.6.6  &   $\leq 1$  & 9.2$\pm$4.4   & $\leq 1$   & 2.0$\pm$0.2 &  $\leq 1$  & * & 2.0$\pm$0 & $\leq 1$ & 2.0$\pm$0.4 & $\leq 1$  & * &2.0$\pm0$  \\
       Asia.8.8  &   $\leq 1$  & 20.7$\pm$4.8  & $\leq 1$   & 11.6$\pm$3.1 & $\leq 1$  & * & 11.0$\pm$2.0 & $\leq 1$ & 2.5$\pm$0.9  & $\leq 1$  & * & 2.1$\pm$0.4   \\
       Bowling.9.11  &   $\leq 1$  & 5.6$\pm$2.6 & $\leq 1$  & 7.0$\pm$2.1 & $\leq 1$  & * & 4.1$\pm$2.3 & $\leq 1$ & 2.8$\pm$1.4 & $2$  & * &2.0$\pm0$ \\
       InsSmall.15.25  &   $\leq 1$   & 39.9$\pm$7.1  & $\leq 1$   & 12.4$\pm$4.9 &  $4$ & *  & 8.3$\pm$1.4 & $\leq 1$ & 24.9$\pm$8.5 & T  & 0.05 &8.1$\pm$1.9    \\
       Rain.14.18  &   $\leq 1$   & 16.7$\pm$3.8   & $\leq 1$  & 3.0$\pm$1.6  &  $1$ & *  & 2.0$\pm0$ & $\leq 1$ & 6.9$\pm$4.8 & $104$  & * &2.0$\pm0$ \\
       Cloud.16.19  &   $\leq 1$   & 43.8$\pm$7.8 & $\leq 1$   & 20.4$\pm$4.0   & $\leq 1$ & *  & 19.0$\pm$4.1 & $\leq 1$ & 4.6$\pm$1.9 & $28$  & * &4.8$\pm$1.1   \\
       Funnel.18.18  &   $\leq 1$  & 15.9$\pm$3.0    & $\leq 1$   & 2.1$\pm$0.3   &  $\leq 1$  & * & 2.0$\pm$0.0 & $\leq 1$ & 4.0$\pm$4.0 & $42$  & * &2.0$\pm0$ \\
       Galaxy.20.22  &   $\leq 1$   & 44.2$\pm$3.9 & $\leq 1$   & 27.9$\pm$6.1  &  $\leq 1$  & * & 6.2$\pm$5.2 & $\leq 1$ & 3.1$\pm$2.8 & $200$  & * &1.0$\pm0$  \\
       Insurance.27.52  &   $\leq 1$  & 43.8$\pm$6.6   & $\leq 1$  & 35.8$\pm$5.7 &  T   & 0.05 & 14.2$\pm$4.7 & $\leq 1$ & 32.5$\pm$12.5 & T  & 0.28 &19.8$\pm$7.5   \\
       Factors.27.68 & $\leq 1$  & 26.1$\pm$5.8 & $\leq 1$   & 46.2$\pm$7.5  &T & 0.03 & 37.8$\pm$5.1 & $\leq 1$ & 72.2$\pm$9.2 & T  & 0.29 &51.1$\pm$8.5  \\
       Hailfinder.56.66  &   4  & 132.6$\pm$19.5    & $\leq 1$  & 57.6$\pm$10.3   &  T  & 0.08 & 23.0$\pm$9.3 & $\leq 1$ & 35.6$\pm$15.9 & T  & 0.19 &10.7$\pm$11.4   \\
       Hepar2.70.123  &   7   & 133.6$\pm$12.7   & 1   & 71.0$\pm$7.6  & T & 0.13  & 42.8$\pm$14.1 & $\leq 1$ & 71.6$\pm$22.8 & T  & 0.32 & 64.2$\pm$17.5  \\
    \end{tabular}}
   \caption*{
\hdbu, high-dimensional bottom-up; TD, top-down; MISOCP, mixed-integer second-order cone program; GES, greedy equivalence search;  $d_{\mathrm{cpdag}}$, differences between the true and estimated completed partially directed acyclic graphs; \rgap, relative optimality gap; T refers to reaching the time limit given by $50m$.}    \label{tab:compare_benchmarks_est}
\end{table}

\subsection{Early Stopping}
\label{subsec:early_stopping}
So far, we have set the desired level of optimality gap to be $\tau = 0$. As described in \S{\ref{sec:early_stopping}} and \S{\ref{sec:theory}}, one can set the optimality gap to $\tau = \mathcal{O}({\lambda^2s^\star })$ to speed up computations while retaining good statistical properties. {Since $s^\star$ is unknown, we can replace it with an upper bound $\bar{s}$ to get the maximum speedup while retaining good statistical properties. A fully dense DAG with $m$ nodes would have $m(m-1)/2$ edges. However, given our assumption that the underlying DAG is sparse, we can instead use $\bar{s}=m(m-1)/4$.} In Fig. ~\ref{fig:early_stopping}(a) , we present the upper and lower bounds of objectives during the solving process for the Galaxy instance. It can be seen that the upper bounds decrease fast at the beginning, indicating we find a solution that is close to optimal very fast. It takes longer for the lower bounds to converge to prove the optimality. The behavior of the upper and lower bounds illustrates how early stopping can reduce computational time without sacrificing statistical accuracy. We next empirically explore how changing $\tau$ in the range $\tau \in \{{0,\lambda^2\bar{s},m\lambda^2\bar{s}}\}$ impacts both the computational and statistical performances of our algorithm. Concretely, we generate a directed acyclic graph among $m$ nodes with $m$ directed edges using \texttt{randomDAG} from \texttt{R} package \texttt{pcalg}. For this network structure, we generate $n=400$ independent and identically distributed samples from \eqref{eq:compact_sem_model} where the nonzero entries of $B^\star$ are drawn uniformly at random from the set $\{-0.8, -0.6, 0.6, 0.8\}$ and the diagonal entries of $\Omega^\star$ are chosen uniformly at random from the set $\{0.5,1,1.5\}$. We then evaluate the performance of our algorithm for $m \in \{{15,20,25,30}\}$ and different values of $\tau$ over $30$ independent datasets. 

Fig.~\ref{fig:early_stopping} (b, c) summarizes the results. Here, the metric $d_{\mathrm{cpdag}}$ is scaled by $m$, the total number of edges in the true underlying directed acyclic graph. Further, time is scaled by $50m$. We observe that applying an early stopping criterion, using $\tau \geq \lambda^2\bar{s}$, leads to notable computational benefits. For example, when $\tau = 0$ and $m = 30$, we cannot solve 29 of the 30 instances within the time limit of $50m$ seconds, whereas early stopping with $\tau = \lambda^2\bar{s}$ finishes within the time limit for 25 instances with a total average time of approximately $730$ seconds across the 30 instances. Furthermore, as predicted by our theoretical results, the faster computation time does not come at a statistical cost when $\tau = \lambda^2\bar{s}$. For example, when $m = 20$, all 30 instances for $\tau = \{0,\lambda^2\bar{s}\}$ are completed and the $d_\mathrm{cpdag}$ metrics are similar. We also observe that when $\tau = m\lambda^2\bar{s}$, the scaled $d_{\mathrm{cpdag}}$ values are substantially larger than those of $\tau = \lambda^2\bar{s}$, further supporting our theoretical results.

\begin{figure}[t]
    \centering
\subfloat[Upper and lower bounds over time]{        \includegraphics[scale=0.33]{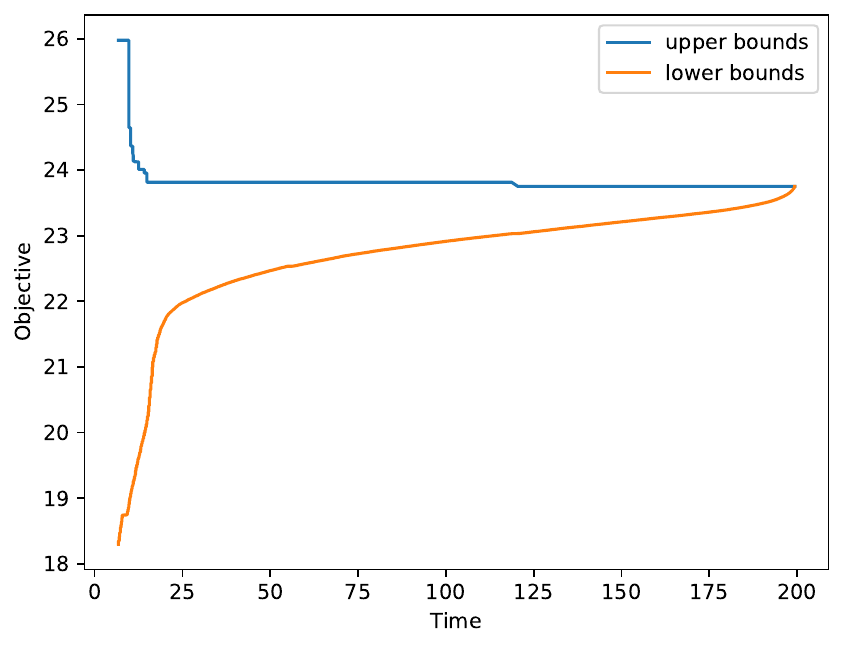}}
  \subfloat[Box plots of scaled $d_{\mathrm{cpdag}}$]{
        \includegraphics[scale=0.33]{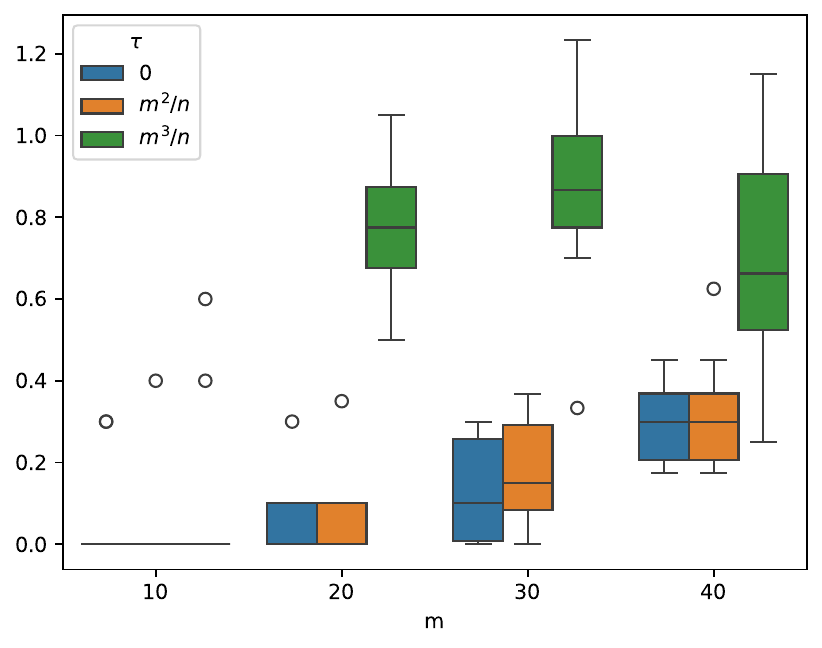}}
    \subfloat[Bar plots of average scaled time]{        \includegraphics[scale=0.33]{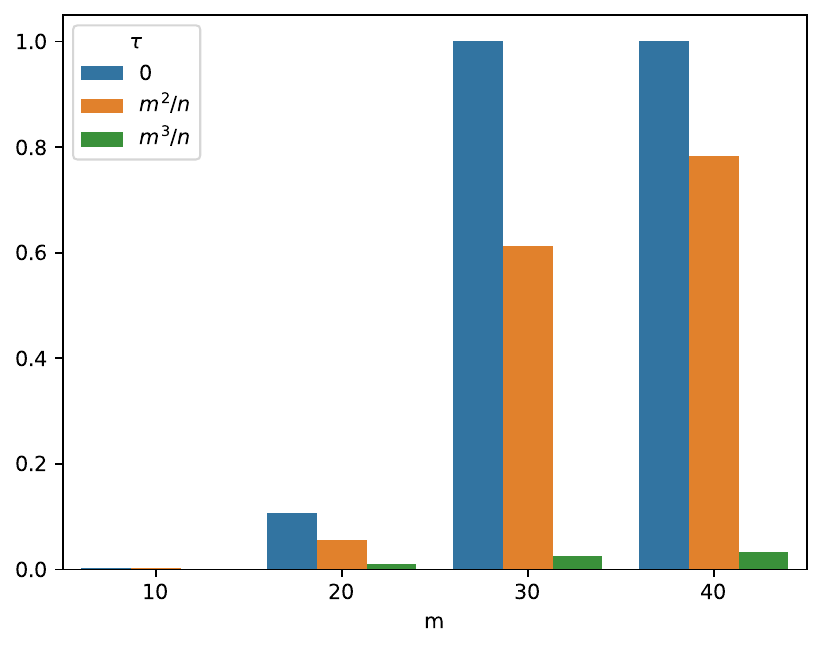}}
    \caption{(a) Upper and lower bounds from the branch-and-bound process for the Galaxy instance. {Scaled $d_{\mathrm{cpdag}}$ (b) and time (c) in early stopping results with $\tau \in\{0,\lambda^2\bar{s},m\lambda^2\bar{s}\}$ across 30 independent trials.}}
    \label{fig:early_stopping}
    
\end{figure}

\subsection{Comparison of General Constraint Attributes Versus Outer Approximation}
\label{subsec:OA}
In \S{\ref{sec:solution method}}, we proposed an outer approximation technique for handling the logarithm term in our mixed integer program \eqref{problem:micp}. We next numerically illustrate the computational and statistical benefits of using our outer approximation approach over \texttt{Gurobi}'s general constraint attribute. Specifically, we generate $30$ independent datasets according to the setup in \S{\ref{subsec:comparison_benchmarks}}. Table \ref{tab:compare_oa_general} compares the performance of solving \eqref{problem:micp} using the outer approximation technique via \texttt{Gurobi} lazy callbacks versus using the default general constraint attribute of \texttt{Gurobi}. {It can be seen that our method with outer approximation performs similarly as the \texttt{Gurobi} default in terms of $d_\text{cpdag}$ while offering computational advantages for instances that can be solved within the time limit.}




\begin{table}[!ht]
\def~{\hphantom{0}}
\caption{{Comparison of \texttt{Gurobi} default with our outer approximation method}}{
    \centering
    \scalebox{0.9}{\begin{tabular}{lcccccc}
        & \multicolumn{3}{c}{\texttt{Gurobi} default} & \multicolumn{3}{c}{Our method}\\
       Network($m$) & Time & \rgap & $d_{\mathrm{cpdag}}$ & Time & \rgap  & $d_{\mathrm{cpdag}}$ \\
       Dsep(6)  &   $3$ & *   & 2.0$\pm 0$  & $\leq 1$ & *     & 2.0$\pm 0$\\
       Asia(8)  &   $6$ & *   & 2.1$\pm 0.4$ &   $\leq 1$  & *    & 2.1$\pm 0.4$ \\
       Bowling(9)  &   $23$ & *   & 2.0$\pm 0$  &   $\leq 1$& *    & 2.0$\pm 0$ \\
       InsSmall(15)  &  T & 0.09     & 7.3$\pm 3.1$ &   T & 0.05& 7.6$\pm 2.2$ \\
       Rain(14)  &   $285$ & *      & 2.0$\pm 0$ &   $88$& *   & 2.0$\pm 0$ \\
       Cloud(16)  &   $119$ & *    & 4.5$\pm 1.1$  &   $11$& *     & 4.6$\pm 1.0$ \\
       Funnel(18)  &   $130$ & *    & 2.0$\pm 0$ &   $62$& *     & 2.1$\pm 0.3$ \\
       Galaxy(20)  &   $764$ & *     & 1.1$\pm 0.5$ &   $114$& *    & 1.0$\pm 0$ \\
       Insurance(27)  &   T & 0.26      & 18.1$\pm 7.6$  &   T  & 0.24& 10.7$\pm 7.0$ \\
       Factors(27) & T & 0.30  & 56.1$\pm 7.2$ & T & 0.30& 47.7$\pm 9.5$\\
       Hailfinder(56)  &   T & 0.20    & 8.6$\pm 5.3$  &   T & 0.23& 14.9$\pm 11.2$ \\
       Hepar2(70)  &   T & 0.30      & 57.3$\pm 9.2$   &   T & 0.29& 55.0$\pm 13.3$\\
    \end{tabular}}}
    \label{tab:compare_oa_general}
 \caption*{
Here, $d_{\mathrm{cpdag}}$, differences between the true and estimated completed partially directed acyclic graphs; \rgap, relative optimality gap; T indicates reaching the time limit of $50m$. All results are averaged across 30 independent trials.}
\end{table}

\subsection{Robustness to Different Amounts of Noise Heteroscedasticity}
\label{subsec:var}
We provide more details of the experiment presented in \S{\ref{sec:our_contributions}}. We consider the empirical setup in \S{\ref{subsec:early_stopping}}, with the exception that the diagonal entries of $\Omega^\star$ are sampled uniformly at random from the interval $[4-\rho,4+\rho]$, where $\rho \in \{1,2,4\}$. Larger values of $\rho$ thus indicate a larger degree of noise heteroscedasticity. We empirically evaluate the performance of our method and competing methods across 30 independent trials. 

Table \ref{tab:var} and Fig. \ref{fig:alpha_results} show that our method is robust to different levels of noise heteroscedasticity, producing small values of $d_\text{cpdag}$ for different $\rho$. This is in contrast to the competing methods that rely on homoscedastic noise: they yield large values of $d_\text{cpdag}$ under strong violation of the homoscedasticity assumption (e.g., $\rho = 4$). 



\begin{table}[!ht]
\def~{\hphantom{0}}
\caption{{Comparison to benchmarks for different amounts of noise heteroscedasticity}}
    \centering
    \resizebox{\columnwidth}{!}{\begin{tabular}{ccccllllcccccccc}
       \multicolumn{1}{c}{} & \multicolumn{1}{c}{}  & \multicolumn{2}{c}{\hdbu} & \multicolumn{2}{c}{\gsp} & \multicolumn{2}{c}{\td}& \multicolumn{3}{c}{\misocp}  & \multicolumn{2}{c}{\ges}  & \multicolumn{3}{c}{\micp}\\
       $m$ & $\rho$ & Time & $d_{\mathrm{cpdag}}$  & Time &$d_{\mathrm{cpdag}}$  & Time &$d_{\mathrm{cpdag}}$ & Time & \rgap & $d_{\mathrm{cpdag}}$  & Time & $d_{\mathrm{cpdag}}$ & Time & \rgap & $d_{\mathrm{cpdag}}$\\
        10 & 1 & $0.09$ & 6.3$\pm4.1$& $0.01$ &2.1$\pm1.7$& $\leq0.01$ &0.4$\pm0.8$& $0.06$ & * & 0.6$\pm1.3$& $0.01$ & 3.4$\pm2.7$& $0.76$   & * & 0.5$\pm0.9$\\
        10 & 2 & $0.08$ & 7.1$\pm4.0$& $0.01$ &2.1$\pm1.8$& $\leq0.01$ &0.4$\pm0.8$& $0.06$ & * & 0.6$\pm1.4$& $\leq 0.01$ & 1.6$\pm1.9$& $0.76$   & * & 0.4$\pm0.9$\\
        10 & 4 & $0.08$ & 5.7$\pm3.2$& $0.01$ &4.2$\pm1.4$& $\leq0.01$ &2.4$\pm1.2$& $0.05$ & * & 0.5$\pm1.4$& $\leq 0.01$ & 1.7$\pm2.2$& $0.66$   & * & 0.5$\pm1.2$\\
        15 & 1 & $0.14$ & 10.3$\pm4.8$& $0.01$ &9.6$\pm2.5$& $0.01$ &3$\pm2.4$& $0.12$ & * & 2.7$\pm2.5$& $\leq 0.01$ & 6.4$\pm3.9$& $7.41$   & * & 1.3$\pm1.9$\\
        15 & 2 & $0.13$ & 11.7$\pm4.3$& $0.01$ &10.1$\pm3.3$& $0.01$ &5.5$\pm2.6$& $0.12$ & * & 5.6$\pm2.9$& $\leq 0.01$ & 5.7$\pm4.3$& $5.17$   & * & 1.3$\pm2.0$\\
        15 & 4 & $0.13$ & 16.4$\pm5.9$& $0.01$ &9.9$\pm2.3$& $\leq0.01$ &11.6$\pm3.3$& $0.10$ & * & 7.4$\pm2.8$& $0.01$ & 6.0$\pm2.7$& $4.60$   & * & 1.8$\pm1.9$\\
        20 & 1 & $0.21$ & 17.9$\pm6.7$& $0.01$ &16.1$\pm3.8$& $0.01$ &6.7$\pm4.3$& $0.21$ & * & 1.9$\pm1.8$& $\leq 0.01$ & 11.1$\pm5.3$& $177.06$ & * & 0.7$\pm1.3$\\
        20 & 2 & $0.21$ & 20.2$\pm5.3$& $0.01$ &16$\pm3.1$& $0.01$ &9.9$\pm3.7$& $0.25$ & * & 4.3$\pm4.5$& $0.01$ & 12.0$\pm4.7$& $68.09$  & * & 1.6$\pm2.8$\\
        20 & 4 & $0.22$ & 28.1$\pm4.7$& $0.01$ &17.3$\pm3.4$& $0.02$ &17.8$\pm4.4$& $0.22$ & * & 10.2$\pm3.9$& $\leq 0.01$ & 10.5$\pm5.5$& $116.20$ & * & 2.6$\pm3.2$\end{tabular}}
    \label{tab:var}
\caption{Here, \hdbu, high-dimensional bottom-up; \gsp, greedy sparsest permutations; \td, top-down;  \misocp, mixed-integer second-order cone program; \ges, greedy equivalence search;  $d_{\mathrm{cpdag}}$, differences between the true and estimated completed partially directed acyclic graphs; \rgap, relative optimality gap. All results are averaged across 30 independent trials.}
\end{table}

\section{Discussion}

We discuss some future research questions that arise from our work. The mixed-integer framework allows for more than one optimal solution to be identified. It would be of interest to explore this feature to identify multiple solutions in the Markov equivalence class. In addition, the outer approximation algorithm can be further enhanced by developing stronger cutting planes that exploit the constraint structure.

\section*{Acknowledgement}
TX and SK are supported, in part, by a National Science Foundation Grant. 
AT is supported by the Royalty Research Fund at the University of Washington.
AS was supported by grants from the National Science Foundation and the National Institutes of Health. 

\bibliographystyle{elsarticle-harv}
\bibliography{main.bib}

\newpage
\appendix
\section*{Appendix}

\subsection{Proof of Proposition~\ref{prop:main}}
\label{proof:prop_main}

\begin{proof}[Proof of Proposition~\ref{prop:main}]

By replacing $\Theta$ with $\Gamma\Gamma^\T$ for $\Gamma$ such that $\Gamma=(I-B)D^{1/2} $, problem \eqref{Problem:original} is equivalent to the following problem:
\begin{subequations}
\begin{align*}
\min\limits_{\Gamma,B,D} & -2\log\det(\Gamma)+\mathrm{tr}(\Gamma\Gamma^\T\hat{\Sigma}_n) + \lambda^2\norm{\Gamma-\mathrm{diag}(\Gamma)}_{\ell_0} \\
\text { s.t. } & \Gamma= (I-B)D^{1/2},\\
    & B \in \mathcal{B}, D \succ 0.
\end{align*} 
\label{Problem:2}
\end{subequations}
Since $B$ has zero diagonal elements and $D$ is diagonal with strict positive diagonal entries, we have $\norm{\Gamma-\mathrm{diag}(\Gamma)}_{\ell_0} = \norm{B}_{\ell_0}$. We do not require $\Gamma$ to be a lower-triangular matrix at this time.

These observations lead to the next claim that allows us to replace the non-convex constraint $\Gamma= (I-B)D^{1/2}$, with convex (linear) constraints.

\begin{lemma}
\label{prop:main_constraint}
$\{\Gamma: \text{There exists } D \succ 0, B \in \mathcal{B}\text{ s.t. } \Gamma=(I-B)D^{1/2}\} = \{\Gamma: \Gamma_{ii}>0, i=1,\ldots,m; \Gamma - \mathrm{diag}(\Gamma) \in \mathcal{B}\}$.
\end{lemma}

\begin{proof}[Proof of Lemma \ref{prop:main_constraint}]

$\subseteq$: One can write $\Gamma = (I-B)D^{1/2}=D^{1/2}-BD^{1/2}$. Since $B$ is a directed acyclic graph and $D$ is a positive definite diagonal matrix, $BD^{1/2}$ is also a directed acyclic graph. Therefore, 
$$\Gamma - \mathrm{diag}(\Gamma) = D^{1/2}-BD^{1/2} - D^{1/2} = -BD^{1/2}$$
is a directed acyclic graph, and we have $\Gamma_{ii} = D^{1/2}_{ii} > 0, \text{for any } i$.

$\supseteq$: If $\Gamma - \mathrm{diag}(\Gamma) \in \mathcal{B}$ and $\Gamma_{ii} >0 , \text{for any } i$, let $D^{1/2} = \mathrm{diag}(\Gamma)$, then $D \succ 0$.

Then, $D^{-1/2}\{\Gamma - \mathrm{diag}(\Gamma)\} = D^{-1/2}\Gamma - I$ is also a directed acyclic graph.
Let $B=I - \Gamma D^{-1/2}$, then we have $(I-B)D^{1/2} = \Gamma$.
\end{proof}

Therefore, Problem \eqref{Problem:original} is equivalent to the following optimization problem:
\begin{subequations}
\begin{align*}
\min\limits_{\Gamma} & -2\log\det(\Gamma)+\mathrm{tr}(\Gamma\Gamma^\T\hat{\Sigma}_n) + \lambda^2\norm{\Gamma-\mathrm{diag}(\Gamma)}_{\ell_0} \\
\text { s.t. } & \Gamma_{ii}>0 \quad(i=1,\ldots,m),\\
    &  \Gamma - \mathrm{diag}(\Gamma) \in \mathcal{B}.
\end{align*} 
\end{subequations}

Furthermore, for any permutation matrix $P$, we have $\det(\Gamma)=\det(P\Gamma P^\T)$. Since $\Gamma - \mathrm{diag}(\Gamma) \in \mathcal{B}$, if $\Gamma_{ij}\not=0$ for $i\not=j$, we deduce that $\Gamma_{ji}$ must be 0. Therefore, there exists a permutation matrix $\Bar{P}$ such that $\Bar{P}\Gamma\Bar{P}^\T$ is a lower-triangular matrix. As a result, $\log\det(\Gamma) = \sum_{i=1}^m  \log \Gamma_{ii}$. Further, the constraint $\Gamma_{ii}>0$ can be removed due to the logarithm barrier function.  

If $(\hat{B}^{\mathrm{opt}},\hat{D}^{\mathrm{opt}})$ is a minimizer of \eqref{Problem:original}, as shown in the proof of Lemma \ref{prop:main_constraint}, we have the corresponding feasible point $\hat{\Gamma}^{\mathrm{opt}} = (I-\hat{B}^{\mathrm{opt}})(\hat{D}^{\mathrm{opt}})^{1/2}$ of \eqref{Problem:3}. Since \eqref{Problem:original} and \eqref{Problem:3} have the same objective function after changing variables, $\hat{\Gamma}^{\mathrm{opt}}$ is also the minimizer of \eqref{Problem:3}. Similarly, if  $\hat{\Gamma}^{\mathrm{opt}}$ is a minimizer of \eqref{Problem:3}, then $\hat{D}^{\mathrm{opt}} = \mathrm{diag}(\hat{\Gamma}^{\mathrm{opt}})^2$ and $\hat{B}^{\mathrm{opt}} = I-\hat{\Gamma}^{\mathrm{opt}}(\hat{D}^{\mathrm{opt}})^{-1/2}$ is the minimizer of \eqref{Problem:original}.
\end{proof}

\subsection{Identifying good superstructures using the graphical lasso}
\label{sec:super_structure}
Our mixed-integer programming framework takes as input a superstructure $E_\mathrm{super}$. Ideally, the true edge set $E^\star$ is contained in this superstructure, i.e., $E^\star \subseteq E_\mathrm{super}$. Moreoever, for computational purposes, we do not want the cardinality of $E_\mathrm{super}$ to be large. We use a graphical lasso estimator to obtain a superstructure with the aforementioned properties. We consider the following graphical lasso estimator \citep{Rothman2008SparsePI}:
$$\hat{\Theta} = \argmin_{\Theta \in \mathbb{R}^{m\times{m}}} -\log\det(\Theta) + \mathrm{tr}(\Theta\hat{\Sigma}_n) + \lambda_\mathrm{glasso}^2\|\Theta-\mathrm{diag}(\Theta)\|_{\ell_1}.$$
A natural candidate for a superstructure is the estimated moral graph $E_\mathrm{super} = \mathrm{support}(\hat{\Theta})$. The moral graph of a directed acyclic graph $\mathcal{G}$ forms an undirected graph of $\mathcal{G}$ by adding edges between nodes that have common children. Furthermore, we let $\Theta^\star = {\Sigma^\star}^{-1}$ be the precision matrix underlying the data. 
\begin{proposition} Suppose Assumptions \ref{assump:sample_cond} and \ref{condition:faithfulness} hold. Furthermore, suppose that every non-zero element of $\Theta^\star$ has magnitude larger than $\mathcal{O}\left[\{(m+\|\Theta^\star\|_{\ell_0})\log m/n\}^{1/2}\right]$. Then, with probability tending to $1$, for $\lambda_\mathrm{glasso}^2 \asymp \log m/n$, $E^\star \subseteq E_\mathrm{super}$.
\label{prop:glasso}
\end{proposition}
The proof relies on the following result by \cite{Rothman2008SparsePI}:
\begin{lemma}[Theorem 1 of \cite{Rothman2008SparsePI}] Suppose $\lambda_\mathrm{glasso}^2 \asymp \log m/n$. Then, with probability tending to $1$ (as $m,n \to \infty$):
$$\|\Theta^\star - \hat{\Theta}\|_F^2 \leq \mathcal{O}\{(m+\|\Theta^\star\|_{\ell_0})\log m/n\}.$$
\label{lemma:rothman}
\end{lemma}
\begin{proof}[Proof of Proposition~\ref{prop:glasso}]
Since $X$ is a Gaussian random vector, zeros in the precision matrix $\Theta^\star$ specify conditional independence relationships. Specifically, $\Theta^\star_{ij} = 0$ if and only if $X_i \perp X_j | X_{\setminus\{i,j\}}$, or equivalently, $\Theta^\star_{ij} \neq 0$ if and only if $X_i \not\perp X_j | X_{\setminus\{i,j\}}$. The aforementioned property of the precision matrix $\Theta^\star$ results in two other important properties. First, if $i \rightarrow k \leftarrow j$ is a v-structure in the graph $\mathcal{G}^\star$, then necessarily $X_i \not\perp X_j | X_{\setminus\{i,j\}}$ which equivalently means that $\Theta_{ij}^\star \neq 0$. 
Second, under faithfulness (Assumption~\ref{condition:faithfulness}), if pairs of nodes $(i,j)$ are connected in the population graph $\mathcal{G}^\star$, then  $X_i \not\perp X_j | X_{\setminus\{i,j\}}$ or equivalently $\Theta^\star_{ij} \neq 0$. The aforementioned properties imply that the sparsity pattern of $\Theta^\star$ encodes the moral graph structure of $\mathcal{G}^\star$. Now from Lemma~\ref{lemma:rothman}, with probability tending to $1$, $\|\Theta^\star - \hat{\Theta}\|_\infty^2 \leq \mathcal{O}\{(m+\|\Theta^\star\|_{\ell_0})\log m/n\}$. Since the non-zero entries of $\Theta^\star$ have sufficiently large magnitudes, we conclude that if $\Theta^\star_{ij} \neq 0$, then $\hat{\Theta}_{ij} \neq 0$. Putting everything together, we conclude that: $E^\star \subseteq \{(i,j): \Theta^\star_{ij} \neq 0\} \subseteq E_\mathrm{super}$.
\end{proof}

Note that Proposition~\ref{prop:glasso} does not guarantee that the cardinality of $\mathrm{E}_\mathrm{super}$ is not large. Under a stringent assumption known as the irrepresentability condition, \cite{Ravikumar2008HighdimensionalCE} show that $\mathrm{support}(\hat{\Theta}) = \mathrm{support}(\Theta^\star)$ so that $\text{moral}(E^\star) = E_\mathrm{super}$ (here $\text{moral}(E^\star)$ is the edge set of the moral graph of $\mathcal{G}^\star = (V^\star,E^\star)$). Since the irrepresentability condition may not often hold, to obtain a low cardinality $\mathrm{E}_\mathrm{super}$, we threshold the entries of $\hat{\Theta}$ at level $\tau$ and let $E_\mathrm{super}$ correspond to the support of the thresholded estimate. It is straightforward to show, by appealing to Lemma~\ref{lemma:rothman}, that if $\tau = \mathcal{O}\left[\{(m+\|\Theta^\star\|_{\ell_0})\log m/n\}^{1/2}\right]$ and the smallest nonzero magnitude entry of $\Theta^\star$ is above $\mathcal{O}\left[\{(m+\|\Theta^\star\|_{\ell_0})\log m/n\}^{1/2}\right]$, then with high probability, $\text{moral}(E^\star) = E_\mathrm{super}$. Thus, if the maximal degree in the moral graph of the underlying directed acylic graph is $\mathcal{O}(1)$, then $d_\mathrm{max} = \mathcal{O}(1)$.

In our experiments, we choose $\lambda_\mathrm{glasso}^2 = \log m/n$ and set $\tau = 0.1$.

\subsection{{Proof of Theorem~\ref{thm:main}}}
\label{proof:thm_1}
We are analyzing the following estimator:
\begin{subequations}
\begin{align}
\min\limits_{\Gamma \in \mathbb{R}^{m \times m}: \Gamma-\mathrm{diag}(\Gamma) \in \mathcal{B}} & \quad \sum_{i=1}^m-2\log \Gamma_{ii}+\mathrm{tr}(\Gamma\Gamma^\T\hat{\Sigma}_n) + \lambda^2\norm{\Gamma-\mathrm{diag}(\Gamma)}_{\ell_0} \label{eq:prop-obj} \\\text{subject-to}&\quad \text{support}(\Gamma-\mathrm{diag}(\Gamma)) \subseteq E_\mathrm{super}~~;~~ \|\Gamma_{:j}\|_{\ell_0}  \leq \tilde{\alpha}n/\log{m}, j=1,\ldots,m,
\end{align} 
\label{Problem:temp}
\end{subequations}
where $E_\mathrm{super}$ satisfies $E^\star \subseteq E_\mathrm{super}$.

\begin{proposition}
    Suppose Assumptions~\ref{condition:condition4}--\ref{assump:sample_cond} hold with constants $\tilde{\alpha}, \eta_0$ sufficiently small. Let $\hat{\Gamma}^{\mathrm{early}}$ be the early stopping solution of the optimization problem \eqref{Problem:temp} with the additional constraint that $\left\|\Gamma_{: j}\right\|_{\ell_0} \leq$ $\tilde{\alpha} n / \log m$ for every $j \in V$. Let $\hat{\pi}^{\mathrm{early}}$ be the the ordering of the variables associated with $\hat{\Gamma}^{\mathrm{early}}$. Let $(\hat{B}^{\mathrm{early}}, \hat{\Omega}^{\mathrm{early}})$ be the associated connectivity and noise variance matrices satisfying $\hat{\Gamma}^{\mathrm{early}}=\left(I-\hat{B}^{\mathrm{early}}\right)\left(\hat{\Omega}^{\mathrm{early}}\right)^{-1 / 2}$. Let $\alpha_0=(4 / m) \wedge 0.05$. Then, for $\lambda^2 \asymp \log m / n$, we have, with probability greater than $1-\alpha_0$,
    \begin{itemize}
        \item[(a)] $
\left\|\hat{B}^{\mathrm{early}}-\tilde{B}^{\star}\left(\hat{\pi}^{\mathrm{early}}\right)\right\|_F^2+\left\|\hat{\Omega}^{\mathrm{early}}-\tilde{\Omega}^{\star}\left(\hat{\pi}^{\mathrm{early}}\right)\right\|_F^2\leq\mathcal{O}\left(\lambda^2 s^{\star}\right)$, and $\left\|\tilde{B}^{\star}(\hat{\pi}^\text{early})\right\|_{\ell_0} \asymp s^{\star}.
$
        \item[(b)] $\hat{\Omega}^\mathrm{early}_{jj} \geq \mathcal{O}(1)$, for $j=1,\ldots,m$.
    \end{itemize}
    
\label{thm:vgb_early}
\end{proposition}

The proof of this proposition is similar to the proof of Theorem 3.1 from \cite{vgBuhlmann} with some modifications as we use early stopping solution as the estimate, and the equivalence between formulations  \eqref{Problem:original} and \eqref{Problem:3} established in Proposition~\ref{prop:main}. To fill the gaps and establish Proposition~\ref{thm:vgb_early}, we first analyze the key property of the early stopping solution in the following lemma and remark.

\begin{lemma}
    For any feasible solution $\Gamma$ to the optimization problem \eqref{Problem:temp}, there exists a diagonal matrix $D$ with positive diagonal entries such that, after applying a transformation $\Gamma \leftarrow \Gamma D^{1/2}$, we have $\mathrm{tr}\{\Gamma\Gamma^\T \hat{\Sigma}_n\} = m$. Moreover, the objective function value \eqref{eq:prop-obj} after the transformation is not higher than the original. 
\label{lem:trace_m}
\end{lemma}
\begin{proof}[Proof of Lemma~\ref{lem:trace_m}]
    Take any feasible $\Gamma$ to the optimization problem \eqref{Problem:temp}. Consider the transformation $\Gamma \leftarrow \Gamma{D}^{1/2}$ where $D$ is a diagonal matrix with positive entries. The transformed $\Gamma$ matrix remains feasible without changing the regularization term in the objective. Thus, the best transformation that yields the smallest objective value can be found via the following optimization model:
    \begin{eqnarray*}
D^\mathrm{opt} = \argmin_{D ~\text{diagonal}} -2\sum_{i=1}^m \log(\Gamma{D}^{1/2}) + \mathrm{tr}(D\Gamma^\top\hat{\Sigma}_n\Gamma).
    \end{eqnarray*}
Based on the optimality condition of the equation above, it is straightforward to show that
$$D^\mathrm{opt}_{ii} = 1/(\Gamma^\top\hat{\Sigma}_n\Gamma)_{ii}.$$
Thus, $\mathrm{tr}(D^\mathrm{opt}\Gamma^\top\hat{\Sigma}_n\Gamma) = m$, or equivalently with the transformed $\Gamma \leftarrow \Gamma({D}^\mathrm{opt})^{1/2}$, we have that: $\mathrm{tr}(\Gamma\Gamma^\top\hat{\Sigma}_n) = m$. 
\end{proof}

\begin{remark}\label{rem:tracem}
    Let $\hat{\Gamma}^\mathrm{early}$ be the early stopping solution. Without loss of generality, we can assume that for any $D \neq I$, the transformation $\hat{\Gamma}^\mathrm{early}{D}^{1/2}$ to the early-stopped solution $\hat{\Gamma}^\mathrm{early}$ leads to a larger objective value compared to when $D=I$. If this is not the case, one can find the best transformation $D^\mathrm{opt}$ and obtain a transformed solution with a smaller objective value. Therefore, without loss of generality, $\mathrm{tr}\{\hat{\Gamma}^\mathrm{early}(\hat{\Gamma}^\mathrm{early})^\T \hat{\Sigma}_n\} = m$. Moreover, by the nature of the optimal solution, $\mathrm{tr}\{\hat{\Gamma}^\mathrm{opt}(\hat{\Gamma}^\mathrm{opt})^\T \hat{\Sigma}_n\} = m$.
\end{remark}

Lemma~\ref{lem:trace_m} demonstrates that the trace term of the objective function is fixed at $m$ for both the early stopping solution and the optimal solution. This property provides a key condition for ensuring that the early stopping solution remains sufficiently close to the optimal solution, thereby guaranteeing the validity of Proposition~\ref{thm:vgb_early}. Leveraging the trace term property established in Lemma~\ref{lem:trace_m}, the following lemma provides specific conditions under which the results of Proposition~\ref{thm:vgb_early} hold.

For some constant $K_0 > 0$, consider the set $\mathcal{T}_0:=\{B \in \mathbb{R}^{m \times m}: K_0^2 \geq \|\mathcal{X}_j - \mathcal{X} {B}_{:j}\|_2^2 / n \geq 1/K_0^2,  j=1,\ldots,m\}$ \citep{vgBuhlmann}, where $\mathcal{X}$ is the $n\times m$ sample matrix.

\begin{lemma}
    Consider the setup in Proposition~\ref{thm:vgb_early}, and let ${\Gamma}$ be any feasible solution of \eqref{Problem:temp} with the corresponding adjacency matrix $B$ lying inside the set $\mathcal{T}_0$ and
    $$
    m -2\log\det({\Gamma}) +\lambda^2\|\Gamma-\mathrm{diag}(\Gamma)\|_{\ell_0}\leq \ell_n(\Gamma^\star{\Gamma^\star}^\T) + \mathcal{O}(\lambda^2 s^\star),
    $$
   where $\Gamma^\star = (I-B^\star){D^\star}^{1/2}$. Then, part (a) of Proposition~\ref{thm:vgb_early} holds for $\Gamma$.
\label{lem:connect_vgb}
\end{lemma}

The proof of Lemma~\ref{lem:connect_vgb} follows from the proof of Theorem 3.1 in \cite{vgBuhlmann}. In particular, Theorem 3.1 of \cite{vgBuhlmann} establishes the closeness of an optimal parameter $(\hat{B}^\mathrm{opt},\hat{D}^\mathrm{opt})$ of \eqref{Problem:original} to population parameters. This result is an analog of the first result of Proposition~\ref{thm:vgb_early}. In their proof, the optimality of $(\hat{B}^\mathrm{opt},\hat{D}^\mathrm{opt})$ only comes in through the inequality:
\begin{eqnarray}
\ell_n\{\hat{\Gamma}^\mathrm{opt}({\hat{\Gamma}^\mathrm{opt}})^\top\} + \lambda^2 \|\hat{\Gamma}^\mathrm{opt}-\mathrm{diag}(\hat{\Gamma}^\mathrm{opt})\|_{\ell_0} \leq \ell_n(\Gamma^\star{\Gamma^\star}^\T) + \lambda^2{s}^\star, 
\label{eqn:basic}
\end{eqnarray}
where $\hat{\Gamma}^\mathrm{opt} = (I-\hat{B}^\mathrm{opt})(\hat{D}^\mathrm{opt})^{1/2}$. From Lemma~\ref{lem:trace_m}, we know that $\mathrm{tr}(\hat{\Gamma}^\mathrm{opt}({\hat{\Gamma}^\mathrm{opt}})^\top\hat{\Sigma}_n) = m$. Thus, the inequality \eqref{eqn:basic} reduces to:
\begin{eqnarray*}
m - 2\log\det(\hat{\Gamma}^\mathrm{opt}) + \lambda^2 \|\hat{\Gamma}^\mathrm{opt}-\mathrm{diag}(\hat{\Gamma}^\mathrm{opt})\|_{\ell_0} \leq \ell_n(\Gamma^\star{\Gamma^\star}^\T) + \lambda^2{s}^\star.
\end{eqnarray*}
Therefore, if the inequality \eqref{eqn:basic} and other conditions in Lemma~\ref{lem:connect_vgb} hold for a feasible solution $\Gamma$, the first result of Proposition~\ref{thm:vgb_early} holds for $\Gamma$.



The following lemma states that the adjacency matrix $B$ associated with the early stopping solution lies within the set $\mathcal{T}_0$.  This result allows us to apply Lemma~\ref{lem:connect_vgb} to establish the desired result.

\begin{lemma}[Theorem 7.3 of \cite{vgBuhlmann}]
     Suppose Assumptions~\ref{assump:sample_cond} and \ref{ass:noise_lower_bound} hold, and that $$ 1/K_1 := 3 \underline{\kappa} / 4-\left\{\frac{2(t+\log m)}{n}\right\}^{1/2}-3 \bar{\kappa} \tilde{\alpha}^{1/2} > 0.
     $$
     For all $t>0$, with probability at least $1-2 \exp (-t)$,
$$
\|\mathcal{X} \beta\|_2/n^{1/2} \geq\|\beta\|_2 / K_1
$$
uniformly in all $\beta \in \mathbb{R}^m$, where $\beta$ is restricted to have at most $\tilde{\alpha}n/\log m$ nonzeros.
\label{lem:vgb_73}
\end{lemma}

\begin{proof}[Proof of Proposition~\ref{thm:vgb_early}]

We first prove part (b) of the proposition. Recall that $\hat{\Gamma}^\mathrm{early} = (I-\hat{B}^\mathrm{early})(\hat{\Omega}^\mathrm{early})^{-1/2}$. Similar to Lemma~\ref{lem:trace_m}, consider the transformation $\hat{\Gamma}^\mathrm{early} \leftarrow \hat{\Gamma}^\mathrm{early}{D}$ for any diagonal matrix with positive entries, or equivalently,  $\hat{\Gamma}^\mathrm{early} \leftarrow (I-\hat{B}^\mathrm{early})\Omega^{-1/2}$ for some diagonal matrix $\Omega$ with positive entries. The transformed $\hat{\Gamma}^\mathrm{early}$ remains feasible without changing the regularization term in the objective. The best transformation that yields the smallest objective can be found via the following optimization model
 \begin{eqnarray*}
\Omega^\mathrm{opt} = \argmin_{\text{diagonal}~\Omega} \sum_{i=1}^m\log({\Omega}_{ii}) + \mathrm{tr}\{\Omega^{-1}(I-\hat{B}^\mathrm{early})^\top\hat{\Sigma}_n(I-\hat{B}^\mathrm{early})\}.
    \end{eqnarray*}
Based on the optimality condition of the optimization problem above, it is straightforward to show that
$$\Omega^\mathrm{opt}_{ii} = \{(I-\hat{B}^\mathrm{early})^\top\hat{\Sigma}_n(I-\hat{B}^\mathrm{early})\}_{ii} = \|\mathcal{X}_i - \mathcal{X}\hat{B}^{\mathrm{early}}_{:i}\|_2^2/n.$$  
Therefore, without loss of generality, we can assume that any $\Gamma = (I-\hat{B}^\mathrm{early})\Omega$ with $\Omega \neq ({\hat{\Omega}^\mathrm{early}})^{-1/2}$ leads to a larger objective value than that of $\hat{\Gamma}^{\mathrm{early}}$. If this is not the case, one can find the best transformation $\Omega^\mathrm{opt}$ and obtained a transformed solution with a smaller objective value. Therefore, we have shown that $\hat{\Omega}^\mathrm{early}_{ii} = \|\mathcal{X}_i - \mathcal{X}\hat{B}^{\mathrm{early}}_{:i}\|_2^2/n$. 

Using standard arguments from Gaussian concentration and the fact that $n \geq \mathcal{O}(m)$, spectral norm of the difference of $\hat{\Sigma}_n$ and $\Sigma^\star$ is bounded by the order $(m/n)^{1/2}$ with probability greater than $1-\alpha_0$. Consequently, $\sigma_\mathrm{min}(\hat{\Sigma}_n) \geq \sigma_\mathrm{min}(\Sigma^\star) - \mathcal{O}\{(m/n)^{1/2}\} \geq \underline{\kappa} - \mathcal{O}\{(m/n)^{1/2}\} \geq 1/K_2$ for some non-negative constant $K_2$.  

Since matrix $\hat{B}^{\mathrm{early}}$ has zero diagonals, we have $\hat{\Omega}^\mathrm{early}_{ii} \geq \min_{\beta\in R^m, \beta_i = 0} \|\mathcal{X}_i - \mathcal{X}\beta\|_2^2/n = \|\mathcal{X}_i - \mathcal{X}_{-i}(\mathcal{X}_{-i}^\top \mathcal{X}_{-i})^{-1}\mathcal{X}_{-i}^\top\mathcal{X}_i\|_2^2/n$, where $\mathcal{X}_{-i}$ is the data matrix $\mathcal{X}$ excluding the $i$th column. Note that invertibality of $\mathcal{X}_{-i}^\top \mathcal{X}_{-i}$ follows from $\sigma_\mathrm{min}(\hat{\Sigma}_n) > 0$. By applying elementary row and column operations to $\mathcal{X}^\top\mathcal{X}$, we obtain
$$\mathcal{X}^\top\mathcal{X} \rightarrow \begin{bmatrix}
    \mathcal{X}_{-i}^\top \mathcal{X}_{-i} & \mathcal{X}_{-i}^\top \mathcal{X}_{i}\\ \mathcal{X}_{i}^\top \mathcal{X}_{-i} & \mathcal{X}_{i}^\top\mathcal{X}_{i}
\end{bmatrix}.$$ Considering the Schur complement of $\mathcal{X}^\top \mathcal{X}$, we have $\mathrm{det}(\mathcal{X}^\top \mathcal{X}) = \mathrm{det}(\mathcal{X}_{-i}^\top \mathcal{X}_{-i})\mathrm{det}(\mathcal{X}_{i}^\top\mathcal{X}_{i} - \mathcal{X}_{i}^\top \mathcal{X}_{-i}(\mathcal{X}_{-i}^\top \mathcal{X}_{-i})^{-1}\mathcal{X}_{-i}^\top \mathcal{X}_{i})=\mathrm{det}(\mathcal{X}_{-i}^\top \mathcal{X}_{-i})(\mathcal{X}_{i}^\top\mathcal{X}_{i} - \mathcal{X}_{i}^\top \mathcal{X}_{-i}(\mathcal{X}_{-i}^\top \mathcal{X}_{-i})^{-1}\mathcal{X}_{-i}^\top \mathcal{X}_{i})$.

Since $ \|\mathcal{X}_i - \mathcal{X}_{-i}(\mathcal{X}_{-i}^\top \mathcal{X}_{-i})^{-1}\mathcal{X}_{-i}^\top\mathcal{X}_i\|_2^2 = \mathcal{X}_{i}^\top\mathcal{X}_{i} - \mathcal{X}_{i}^\top \mathcal{X}_{-i}(\mathcal{X}_{-i}^\top \mathcal{X}_{-i})^{-1}\mathcal{X}_{-i}^\top \mathcal{X}_{i}$, we have
$$
\|\mathcal{X}_i - \mathcal{X}_{-i}(\mathcal{X}_{-i}^\top \mathcal{X}_{-i})^{-1}\mathcal{X}_{-i}^\top\mathcal{X}_i\|_2^2 = \frac{\mathrm{det}(\mathcal{X}^\top \mathcal{X})}{\mathrm{det}(\mathcal{X}_{-i}^\top \mathcal{X}_{-i})} = \frac{\Pi_{i=1}^m\sigma_i(\mathcal{X}^\top \mathcal{X})}{\Pi_{i=1}^m\sigma_i(\mathcal{X}_{-i}^\top \mathcal{X}_{-i})},
$$
where $\sigma_i(\mathcal{X}^\top \mathcal{X})$ represents the $i$th largest singular value of the matrix $\mathcal{X}^\top \mathcal{X}$.

By Cauchy Interlacing Theorem, we have that $\sigma_j(\mathcal{X}^\top \mathcal{X}) \geq \sigma_j(\mathcal{X}_{-i}^\top \mathcal{X}_{-i})$, $j=1,\ldots, m-1$. Consequently, $\|\mathcal{X}_i - \mathcal{X}_{-i}(\mathcal{X}_{-i}^\top \mathcal{X}_{-i})^{-1}\mathcal{X}_{-i}^\top\mathcal{X}_i\|_2^2/n \geq \sigma_\mathrm{min}(\mathcal{X}^\top \mathcal{X}/n)$. That is, $\hat{\Omega}^\mathrm{early}_{ii}$ is lower bounded by the smallest eigenvalue of the sample covariance matrix $\hat{\Sigma}_n$. Since $\sigma_\mathrm{min}(\hat{\Sigma}_n) \geq 1/K_2$, we have that $\hat{\Omega}^\mathrm{early}_{ii} \geq 1/{K_2}$, which completes the proof of part (b) of Proposition~\ref{thm:vgb_early}.

Now we conclude the proof of part (a)  of the proposition.
Lemma~\ref{lem:vgb_73} implies that for every $j = 1,\ldots,m$, with high probability,
$$ \|\mathcal{X}_j - \mathcal{X}\hat{B}^{\mathrm{early}}_{:j}\|_2^2/n=
 \|\mathcal{X}\hat{B}^{\mathrm{early}-}_{:j}\|_2^2 / n \geq \|\hat{B}^{\mathrm{early}-}_{:j}\|_2^2 / K_1^2 \geq 1/K_1^2,
$$
where we define $\hat{B}^{\mathrm{early}-}_{kj} = -\hat{B}^{\mathrm{early}}_{kj}$ for $k\not=j$ and $\hat{B}^{\mathrm{early}-}_{jj} = 1$. Defining $K_0 = \max\{K_1, K_2\}$, we have $K_0^2 \geq \|\mathcal{X}_j - \mathcal{X}\hat{B}^{\mathrm{early}}_{:j}\|_2^2/n \geq 1/K_0^2$. Thus, with high probability, $\hat{B}^{\mathrm{early}}\in\mathcal{T}_0$. By Lemma~\ref{lem:trace_m} and Remark~\ref{rem:tracem}, we have 
$$
\mathrm{tr}\{\hat{\Gamma}^{\mathrm{early}}(\hat{\Gamma}^{\mathrm{early}})^\T \hat{\Sigma}_n\} - \log\det\{\hat{\Gamma}^{\mathrm{early}}(\hat{\Gamma}^{\mathrm{early}})^\T\} = m - \log\det\{\hat{\Gamma}^{\mathrm{early}}(\hat{\Gamma}^{\mathrm{early}})^\T\}.
$$

The early stopping criterion ensures that the optimality gap is on the order of $\lambda^2 s^\star$. Thus, the objective value of the early stopped solution is at most $\mathcal{O}(\lambda^2s^\star)$ larger than the objective value of the optimal solution, and we have 
\begin{align*}
\ell_n\{\hat{\Gamma}^{\mathrm{early}}(\hat{\Gamma}^{\mathrm{early}})^\top\} + \lambda^2 \hat{s}^{\mathrm{early}} &\leq 
 \ell_n\{\hat{\Gamma}^{\mathrm{opt}}(\hat{\Gamma}^{\mathrm{opt}})^\top\} + \lambda^2 \hat{s}^{\mathrm{opt}} + \mathcal{O}(\lambda^2s^\star)\\
 &\leq \ell_n(\Gamma^\star{\Gamma^\star}^\top) + \lambda^2 s^\star + \mathcal{O}(\lambda^2s^\star).
 \end{align*}
By appealing to Lemma~\ref{lem:connect_vgb}, we complete the proof of part (a) of Proposition~\ref{thm:vgb_early}.   
\end{proof}

\begin{proof}[Proof of Theorem~\ref{thm:main}] 

Next, we combine the above results to prove Theorem~\ref{thm:main}. For notational simplicity, we let $\Omega = \tilde{\Omega}^\star(\hat{\pi}^\mathrm{early})$ and $B = \tilde{B}^\star(\hat{\pi}^\mathrm{early})$. First the mean-value theorem implies that, for any $j \in \{1,\ldots, m\}$, there exists some $\omega$ in the interval between $\hat{\Omega}^{\mathrm{early}}_{jj}$ and $\Omega_{jj}$ such that:
\begin{eqnarray*}
\left\{\Omega^{-1/2}_{jj}-(\hat{\Omega}^{\mathrm{early}})^{-1/2}_{jj}\right\}^2 \leq \frac{1}{4}(\Omega_{jj}-\hat{\Omega}^{\mathrm{early}}_{jj})^2\omega^{-3}. 
\end{eqnarray*}
Note that $\omega \geq \min_{j} \Omega_{jj} - \|\Omega-\hat{\Omega}^{\mathrm{early}}\|_F \geq \min_{j}\Omega_{jj} - \lambda{s^\star}^{1/2} \geq \min_{j}\Omega_{jj}/2$, where the final equality follows from Assumption~\ref{ass:noise_lower_bound}. Putting things together, we have that:
\begin{eqnarray}
\|\Omega^{-1/2}-(\hat{\Omega}^{\mathrm{early}})^{-1/2}\|_F^2 \leq  \frac{2\|\Omega-\hat{\Omega}^{\mathrm{early}}\|_F^2}{\min \Omega_{jj}^3}.
\label{eqn:omega}
\end{eqnarray}
Then,
\begin{eqnarray*}
\begin{aligned}
\|\hat{\Gamma}^{\mathrm{early}}-\tilde{\Gamma}^{\star}\left(\hat{\pi}^{\mathrm{early}}\right)\|_F^2 = &\|(I-\hat{B}^{\mathrm{early}})(\hat{\Omega}^{\mathrm{early}})^{-1/2} - (I-B){\Omega}^{-1/2}\|_F^2 \\
&= \|(I-\hat{B}^{\mathrm{early}})(\hat{\Omega}^{\mathrm{early}})^{-1/2} - (I-B)(\hat{\Omega}^{\mathrm{early}})^{-1/2} \\& \quad + (I-B)(\hat{\Omega}^{\mathrm{early}})^{-1/2} - (I-B){\Omega}^{-1/2}\|_F^2 \\
&=\|(B-\hat{B}^{\mathrm{early}})\{(\hat{\Omega}^{\mathrm{early}})^{-1/2}-{\Omega}^{-1/2}\}+ (B-\hat{B}^{\mathrm{early}}){\Omega}^{-1/2} \\& \quad + (I-B)\{(\hat{\Omega}^{\mathrm{early}})^{-1/2}-{\Omega}^{-1/2}\}\|_F^2\\
&\leq \|\hat{B}^\mathrm{early}-B\|_F^2\{\|\Omega^{-1/2}\|_\infty^2 + \|(\hat{\Omega}^{\mathrm{early}})^{-1/2}\|_\infty^2\} + \|\hat{B}^\mathrm{early}-B\|_F^2\|\Omega^{-1/2}\|_\infty^2 \\
&+  \frac{2\max_{j}\|(I-B)_{j:}\|_2^2 \|\Omega-\hat{\Omega}^\mathrm{early}\|_F^2}{\min\{1,\min_j (\Omega_{jj})^3\}}.\\
    &\leq \mathcal{O}(\lambda^2s^\star).
\end{aligned}
\end{eqnarray*}
The last term of the first inequality follows from relation \eqref{eqn:omega}. For the last inequality, we use Proposition~\ref{thm:vgb_early}, which states that $\|\hat{B}^\mathrm{early}-B\|_F^2 \leq\mathcal{O}(\lambda^2s^\star)$ and $\|\hat{\Omega}^\mathrm{early}-\Omega\|_F^2 \leq\mathcal{O}(\lambda^2s^\star)$. Further, from Proposition~\ref{thm:vgb_early}, we also have $\|(\hat{\Omega}^{\mathrm{early}})^{-1/2}\|_\infty \leq \mathcal{O}(1)$. Now we argue that $\|\Omega^{-1/2}\|_\infty \leq \mathcal{O}(1)$ and $\max_{j}\|(I-B)_{j:}\|_2 \leq \mathcal{O}(1)$ and $\min_j\Omega_{j j} \geq \mathcal{O}(1)$.  To see this, recall that $\Theta^\star = (I-B)\Omega^{-1}(I-B)^\top$ and thus $\Theta^\star_{jj} = \|(I-B)_{j:}\|_2^2\Omega_{jj}^{-1}$, for $j=1,\ldots, m$. By Assumption~\ref{assump:sample_cond}, $\Theta^\star_{jj} = \mathcal{O}(1)$. Since $\|(I-B)_{j:}\|_2^2 \geq 1$, we must then have that $\Omega_{jj}^{-1} \leq \mathcal{O}(1)$. Again, by Assumption~\ref{assump:sample_cond}, we have that $\Omega_{jj}^{-1} \geq \mathcal{O}(1)$, which allows us to conclude that $\|(I-B)_{j:}\|_2 \leq \mathcal{O}(1)$. Putting everything together, we get the last inequality that $\|\hat{\Gamma}^{\mathrm{early}}-\tilde{\Gamma}^{\star}\left(\hat{\pi}^{\mathrm{early}}\right)\|_F^2 \leq \mathcal{O}(\lambda^2s^\star)$, as desired. \end{proof}

\subsection{{Proof of Theorem~\ref{corr:equiv}}}
\label{proof:thm_2}
The proof of the theorem relies on the following lemma. 
\begin{lemma} Consider the optimization problem:
\begin{eqnarray}
\hat{\Theta} = \argmin_{\Theta \in \mathbb{R}^{m \times m},\mathrm{support}(\Theta) \subseteq E_\mathrm{super}} f(\Theta),
\label{eqn:opt_const}
\end{eqnarray}
where $f(\Theta):= -\log\det(\Theta) + \mathrm{tr}(\Theta\hat{\Sigma}_n)$ if $\Theta \succ 0$ and $f(\Theta) = \infty$, otherwise. Let $\Theta^\star = {\Sigma^\star}^{-1}$. Then, with probability greater than $1-1/m$, we have that $\|\hat{\Theta}-\Theta^\star\|_\infty \leq \min\left\{\frac{\sigma_{\max}(\Theta^\star)}{2d_{\max}},\left\{\frac{4(1-c)\log m}{n\sigma_{\max}(\Theta^\star)^2}\right\}^{1/2}\right\}$, where $\sigma_{\max}(\Theta^\star)$ denotes the largest singular value of $\Theta^\star$.
\label{lemma:temp_const}
\end{lemma}

\begin{proof}[Proof of Lemma~\ref{lemma:temp_const}]
Since the objective of $f(\Theta)$ is strictly convex in its domain and the constraint is convex, the solution to \eqref{eqn:opt_const} is unique. The Karush-Kuhn-Tucker (KKT) conditions state that there exists $Z \in \mathbb{R}^{m \times m}$ with $\text{support}(Z) \subseteq E_\mathrm{super}^c$ (here, $E_\mathrm{super}^c$ is the complement of $E_\mathrm{super}$) such that
\begin{eqnarray*}-\hat{\Theta}^{-1} + \hat{\Sigma}_n + Z = 0,
\end{eqnarray*}
where $Z$ is the subdifferential term representing the effect of the support constraint.

Letting $\Delta = \hat{\Theta}-\Theta^\star$, by Taylor series expansion, 
\begin{eqnarray*}\hat{\Theta}^{-1} = (\Theta^\star + \Delta)^{-1} = {\Theta^\star}^{-1} -{\Theta^\star}^{-1}\Delta{\Theta^\star}^{-1} + \mathcal{R}(\Delta),
\end{eqnarray*}
where $\mathcal{R}(\Delta) = {\Theta^\star}^{-1}\sum_{k=2}^\infty (-\Delta\Theta^\star)^{k}$. Letting $\mathbb{I}^\star = {\Theta^\star}^{-1} \otimes {\Theta^\star}^{-1}$, the KKT condition can be restated as
\begin{eqnarray*}
\mathcal{P}_{\Psi}\{\mathbb{I}^\star(\Delta) - \mathcal{R}(\Delta) + \hat{\Sigma}_n-\Sigma^\star\} = 0,
\end{eqnarray*}
where $\Psi$ is the subspace $\Psi := \{M \in \mathbb{R}^{m \times m} : \mathrm{support}(M) \subseteq E_\mathrm{super}\}$ and $\mathcal{P}_\Psi$ is a projection operator onto the subspace $\Psi$ that zeros out entries outside of the support set $E_\mathrm{super}$. We now use Brouwer's fixed point theorem to bound $\|\Delta\|_\infty$. Specifically, we define the following operator $\mathcal{J} : \Psi \to \Psi$ as
\begin{eqnarray*}
\mathcal{J}(\delta) = \delta - (\mathcal{P}_\Psi\mathbb{I}^\star\mathcal{P}_\Psi)^{-1}[\mathcal{P}_{\Psi}\{\mathbb{I}^\star\mathcal{P}_{\Psi}\delta - \mathcal{R}(\delta) + \hat{\Sigma}_n-\Sigma^\star\}].
\end{eqnarray*}
Here, the operator $\mathcal{P}_\Psi\mathbb{I}^\star\mathcal{P}_\Psi$ is invertible by Assumption~\ref{assump:sample_cond}. Notice that any fixed point $\delta$ of $\mathcal{J}$ satisfies the KKT condition \eqref{eqn:opt_const}. However, since the optimization problem \eqref{eqn:opt_const} has a unique minimizer, and KKT conditions are necessary and sufficient for optimality, we have that $\mathcal{J}$ has a unique fixed point. Now consider the following compact set: $\mathcal{B}_r = \{\delta \in \Psi: \|\delta\|_\infty \leq r\}$ with $r = \min\left\{2/\kappa_2\|\hat{\Sigma}_n-\Sigma^\star\|_\infty,1/(3\kappa_1{d}_{\max}),3\kappa_2/(4\kappa_1^3d_{\max}^2)\right\}$, where  $\kappa_1$ and $\kappa_2$ are defined in Assumption~\ref{assump:sample_cond}, and $d_{\max}$ is used in Assumption~\ref{condition:strong_beta}. Consider any $\delta \in \mathcal{B}_r$. Using the sub-multiplicative property of the norm $\vertiii{\cdot}$, we have
\begin{eqnarray*}
\vertiii{{\Theta^\star}^{-1}\delta} \leq \vertiii{\Sigma^\star}\vertiii{\delta} \leq \kappa_{1}d_{\max}\|\delta\|_\infty < 1/3.
\end{eqnarray*}
Note that
\begin{eqnarray*}
\|\mathcal{R}(\delta)\|_\infty \leq \|\Sigma^\star\|_\infty \sum_{k=2}^\infty \vertiii{\Sigma^\star\delta}^k \leq \kappa_1\sum_{k=2}^{\infty} (\kappa_1{d}_{\max}r)^k = \frac{\kappa_1(\kappa_1{d}_{\max}r)^2}{1-\kappa_1{d}_{\max}r} \leq \frac{2}{3}{\kappa}_1^3{d}_{\max}^2r^2 \leq \frac{r{\kappa_2}}{2},
\end{eqnarray*}
where the last inequality follows from the bound on $r$. This allows us to conclude that
\begin{eqnarray*}
\|\mathcal{J}(\delta)\|_\infty \leq \frac{1}{\kappa_2}\{\|\mathcal{R}(\delta)\|_\infty + \|\hat{\Sigma}_n-\Sigma^\star\|_\infty\} \leq r.
\end{eqnarray*}
In other words, we have shown that $\mathcal{J}$ maps $\mathcal{B}_r$ to itself. Appealing to Brouwer's fixed point theorem, we have that the fixed point must also lie inside $\mathcal{B}_r$. Thus, $\|\Delta\| \leq r$. From standard Gaussian concentration results, we have that for any constant $\tilde{c}$, if $n \geq \mathcal{O}(\log m)$ (where the order here depends on $\tilde{c}$), with probability greater than $1-1/m$, $\|\hat{\Sigma}_n-\Sigma^\star\|_\infty \leq \tilde{c}(\log m/n)^{1/2}$. Since $r \leq 2/\kappa_2\|\hat{\Sigma}_n-\Sigma^\star\|_\infty$, $\sigma_{\max}(\Theta^\star) = \mathcal{O}(1)$ by Assumption~\ref{assump:sample_cond}, and $d_{\max} \leq \mathcal{O}\{(n/\log m)^{1/2}\}$ by the theorem statement, we have the desired result. 
\end{proof}

\begin{proof}[Proof of Theorem~\ref{corr:equiv}] For notational ease, we introduce the short-hand notation $\ell_n(\Gamma):= \ell_n(\Gamma\Gamma^\top)$. First, using the same arguments as in the proof of Theorem~\ref{thm:main}, we obtain
$$\|\hat{\Gamma}^{\mathrm{early}} - \tilde{\Gamma}^\star(\hat{\pi}^\mathrm{early})\|_F^2 = \mathcal{O}\left(\lambda^2{s}^\star\right).$$
Since we are considering a $\lambda^2$ with a larger order term than that in Theorem~\ref{thm:main}, a stronger assumption is required on the nonzero entries of  $\tilde{B}^\star(\pi)$. Specifically, to ensure the same result holds, we impose Assumption~\ref{condition:strong_beta} (the strong beta-min condition),   which requires the nonzero entries to be of the same order as $\lambda^2$.

For simplicity, denote $\tilde{\Gamma}^\star = \tilde{\Gamma}^\star(\hat{\pi}^\mathrm{early})$.
We have shown that with probability greater than $1-2\alpha_0$, $\|\hat{\Gamma}^\text{early}-\tilde{\Gamma}^\star\|_F^2 \leq \mathcal{O}(\lambda^2{s}^\star)$. Thus, with probability greater than $1-2\alpha_0$, $\|\hat{\Gamma}^\text{early}-\tilde{\Gamma}^\star\|_\infty\leq \mathcal{O}(\lambda{s^\star}^{1/2})$. Combining Assumptions \ref{ass:noise_lower_bound} and \ref{condition:strong_beta}, we have that with probability greater than $1-2\alpha_0$, the magnitude of every non-zero off-diagonal entry of $\tilde{\Gamma}^\star$ is lower-bounded by $\mathcal{O}(\lambda{s^\star}^{1/2})$. Thus, it follows that with probability greater than $1-2\alpha_0$, $\text{support}(\hat{\Gamma}^{\mathrm{early}}) \supseteq \text{support}(\tilde{\Gamma}^\star)$ and consequently $\|\hat{\Gamma}^{\mathrm{early}}\|_{\ell_0}\geq \|\tilde{\Gamma}^\star\|_{\ell_0}$. By the faithfulness condition in Assumption~\ref{condition:faithfulness}, we know that the sparsest graphs that are compatible with the population distribution are the ones in $\mathrm{MEC}(\mathcal{G}^\star)$. Thus, if we show that $\|\hat{\Gamma}^{\mathrm{early}}\|_{\ell_0} = s^\star$, where $s^\star$ represents the number of edges in the true graph, we would have that $\tilde{\Gamma}^\star$ is a member of the population Markov equivalence class and that $\hat{\Gamma}^{\mathrm{early}}$ specifies the same graph as $\tilde{\Gamma}^\star$. 

Notice that by the faithfulness condition, $\|\tilde{\Gamma}^\star\|_{\ell_0} \geq s^\star$. Suppose as a point of contradiction that $\|\hat{\Gamma}^{\mathrm{early}}\|_{\ell_0} > s^\star$. By early stopping criterion, which ensures that the objective value of the early stopped solution is at most $c\lambda^2$ (for some $0<c<1$) larger than the objective value of the optimal solution, we have 
$$
\ell_n(\hat{\Gamma}^{\mathrm{early}}) + \lambda^2 \hat{s}^{\mathrm{early}} \leq 
 \ell_n(\hat{\Gamma}^{\mathrm{opt}}) + \lambda^2 \hat{s}^{\mathrm{opt}} + c\lambda^2 \leq \ell_n(\Gamma^\star_\mathrm{mec}) + \lambda^2 s^\star + c\lambda^2.
 $$
Re-arranging terms, we obtain
\begin{eqnarray*}
\ell_n(\hat{\Gamma}^{\mathrm{early}}) - \ell_n(\Gamma^\star_\mathrm{mec}) \leq \lambda^2 (s^\star-\hat{s}^{\mathrm{early}}) + c\lambda^2,
\end{eqnarray*}
where by the assumption of our proof, $\hat{s}^\mathrm{early} \geq s^\star+1$. Then, we would reach a contradiction if we show that $\ell_n(\hat{\Gamma}^{\mathrm{early}}) - \ell_n(\Gamma^\star_\mathrm{mec}) \geq -(1-c)d_{\max}m\log m/n$.

Notice that
\begin{eqnarray*}
\ell_n(\hat{\Gamma}^{\mathrm{early}}) \geq -\log\det(\hat{\Theta}) + \mathrm{tr}(\hat{\Theta}\hat{\Sigma}_n):= f(\hat{\Theta}).
\end{eqnarray*}
Since $\Theta^\star = \Gamma^\star_\mathrm{mec}{\Gamma^\star_\mathrm{mec}}^\top$, we have that $\ell_n(\Gamma^\star_\mathrm{mec}) = f(\Theta^\star)$. Thus, in order to show that $\ell_n(\hat{\Gamma}^{\mathrm{early}}) - \ell_n(\Gamma^\star_\mathrm{mec}) \geq -(1-c)d_{\max}m\log m/n$, it suffices to show that $|f(\hat{\Theta}) - f(\Theta^\star)| \leq (1-c)d_{\max}m\log m/n$, since $\ell_n(\hat{\Gamma}^{\mathrm{early}})> f(\hat{\Theta})$.

The KKT conditions of \eqref{eqn:opt_const} state that there exists a dual variable $Z\in\mathbb{R}^{m \times m}$ with $\text{support}(Z) \subseteq E_\mathrm{super}^c$ where
$$ 0 =\nabla{f}(\hat{\Theta})+ Z.$$
A Taylor series expansion yields
$$f(\hat{\Theta})-f(\Theta^\star) = \mathrm{vec}\{\nabla{f}(\hat{\Theta})\}^\top \mathrm{vec}(\hat{\Theta}-\Theta^\star) + \frac{1}{2}\mathrm{vec}(\hat{\Theta}-\Theta^\star)^\top \nabla^2f(\bar{\Theta})\mathrm{vec}(\hat{\Theta}-\Theta^\star),$$
where, for every $i,j \in \{1,\dots,p\}$, $\bar{\Theta}_{ij}$ lies between $\hat{\Theta}_{ij}$ and $\Theta^\star_{ij}$. Since $\mathrm{support}(\hat{\Theta})$ and $\mathrm{support}(\Theta^\star)$ are both contained inside $E_\mathrm{super}$, we have that
$$f(\hat{\Theta})-f(\Theta^\star) = \mathrm{vec}\{\nabla{f}(\hat{\Theta})+Z\}^\top \mathrm{vec}(\hat{\Theta}-\Theta^\star) + \frac{1}{2}\mathrm{vec}(\hat{\Theta}-\Theta^\star)^\top \nabla^2f(\bar{\Theta})\mathrm{vec}(\hat{\Theta}-\Theta^\star).$$
Here, $\mathrm{vec}(\cdot)$ vectorizes an input matrix. From the optimality condition of \eqref{eqn:opt_const}, we have 
$$f(\hat{\Theta})-f(\Theta^\star) = \frac{1}{2}\mathrm{vec}(\hat{\Theta}-\Theta^\star)^\top \nabla^2f(\bar{\Theta})\mathrm{vec}(\hat{\Theta}-\Theta^\star),$$
which yields the following bound
\begin{eqnarray}
\begin{aligned}
        |f(\hat{\Theta})-f(\Theta^\star)| \leq \frac{\|\hat{\Theta}-\Theta^\star\|_F^2}{\sigma_\mathrm{min}\{\nabla^2f(\bar{\Theta})\}} \leq \frac{md_{\max}\|\hat{\Theta}-\Theta^\star\|_\infty^2}{\sigma_\mathrm{min}\{\nabla^2f(\bar{\Theta})\}},
\end{aligned}
\label{eqn:temp_bound}
\end{eqnarray}
where $\sigma_\mathrm{min}(\cdot)$ computes the minimum singular value. Here, the second inequality follows from the fact that $\|\hat{\Theta}-\Theta^\star\|_F \leq {m}^{1/2}\max_{i}\|\hat{\Theta}_{i:}-\Theta^\star_{i:}\|_2 \leq (md_{\max})^{1/2}\max_{i}\|\hat{\Theta}_{i:}-\Theta^\star_{i:}\|_\infty =(md_{\max})^{1/2} \|\hat{\Theta}-\Theta^\star\|_\infty$. Notice that $\nabla^2f(\bar{\Theta}) = \bar{\Theta}^{-1} \otimes \bar{\Theta}^{-1}$. Then, $\sigma_{\min}\{\nabla^2f(\bar{\Theta})\} = 1/\sigma_{\max}(\bar{\Theta})^2$. Further,
\begin{eqnarray*}
\sigma_{\max}(\bar{\Theta}) \geq \sigma_{\max}(\Theta^\star) - \|\Theta^\star-\bar{\Theta}\|_2 \geq \sigma_{\max}(\Theta^\star) - d_{\max}\|\Theta^\star-\bar{\Theta}\|_\infty \geq \sigma_{\max}(\Theta^\star) - d_{\max}\|\Theta^\star-\hat{\Theta}\|_\infty,
\end{eqnarray*}
where the second inequality follows from the fact that $\Theta^\star-\bar{\Theta} \in \Psi$ and that for a matrix $N$ with at most $s$ nonzeros in any row or column, $\|N\|_2 \leq s\|N\|_\infty$. From Lemma~\ref{lemma:temp_const}, we then have that, with probability greater than $1-1/m$, $\sigma_{\max}(\bar{\Theta}) \geq \sigma_{\max}(\Theta^\star)/2$. Plugging that into the bound \eqref{eqn:temp_bound} and again appealing to Lemma~\ref{lemma:temp_const}, we have that $(md_{\max}\|\hat{\Theta}-\Theta^\star\|_\infty^2)/\sigma_\mathrm{min}\{\nabla^2f(\bar{\Theta})\} \leq (1-c)m\log m/n$.
\end{proof}

\subsection{Comparing statistical rates of early-stopped solution and optimal solution}
\label{sec:comparison_early_stopping_theory}
In Theorem~3.1 of \cite{vgBuhlmann}, the authors characterize the statistical properties of formulation  \eqref{Problem:original}. Due to the equivalence in Proposition~\ref{prop:main}, the following result is directly implied by Theorem 3.1 of \cite{vgBuhlmann}. 
\begin{proposition} Suppose Assumptions~\ref{condition:condition4}--\ref{condition:faithfulness} hold with constants $\tilde{\alpha}, \eta_0$ sufficiently small. Let $\hat{\Gamma}^\mathrm{opt}$ be the optimal solution of the optimization problem \eqref{Problem:temp} with the additional constraint that $\|\Gamma_{:j}\|_{\ell_0} \leq \tilde{\alpha}n^{1/2}/\log(m)$ for every $j \in V$. Let $\hat{\pi}^\mathrm{opt}$ be the the ordering of the variables associated with $\hat{\Gamma}^\mathrm{opt}$. Let $(\hat{B}^{\mathrm{opt}},\hat{\Omega}^\mathrm{opt})$ be the associated connectivity and noise variance matrix satisfying $\hat{\Gamma}^\mathrm{opt} = (I-\hat{B}^\mathrm{opt})({{\hat{\Omega}^\mathrm{opt}}})^{-1/2}$. Let $\alpha_0 = (4/m) \land 0.05$. Then for $\lambda^2 \asymp {\log(m)}/{n}$, we have, with probability greater than $1-\alpha_0$,
 \begin{itemize}
        \item $
\|\hat{B}^{\mathrm{opt}} - \tilde{B}^\star(\hat{\pi}^\mathrm{opt})\|_F^2 + \|\hat{\Omega}^{\mathrm{opt}} - \tilde{\Omega}^\star(\hat{\pi}^\mathrm{opt})\|_F^2 = \mathcal{O}(\lambda^2s^\star)\leq\mathcal{O}\left(\lambda^2 s^{\star}\right)$, and $\left\|\tilde{B}^{\star}(\hat{\pi}^\text{opt})\right\|_{\ell_0} \asymp s^{\star}.
$
 \item $\hat{\Omega}^\mathrm{opt}_{jj} \geq \mathcal{O}(1)$, for $j=1,\ldots,m$.
    \end{itemize}
\label{thm:vgb_basic}
\end{proposition}
Then, one can establish the following two results. 
\begin{theorem} Consider the setup in Theorem~\ref{thm:main}. Then, for $\lambda^2 \asymp \log m/n$,  letting $\hat{\pi}^\mathrm{opt}$ be the ordering associated with $\hat{\Gamma}^\mathrm{opt}$, we have that, with probability greater than $1-2\alpha_0$, $\|\hat{\Gamma}^\mathrm{opt}-\tilde{\Gamma}^{\star}(\hat{\pi}^\mathrm{opt})\|_F^2 \leq \mathcal{O}(\lambda^2s^\star)$, and $\|\tilde{\Gamma}^{\star}(\hat{\pi}^\mathrm{opt})\|_{\ell_0} \asymp s^\star$.
\label{thm:temp_1}
\end{theorem}
The proof of this Theorem follows from combining the result of Proposition~\ref{thm:vgb_basic} with the analysis in proof of Theorem~\ref{thm:main} that connects the parameter matrices $B$ and $\Omega$ to $\Gamma$. 
\begin{theorem}
Consider the setup in Theorem~\ref{corr:equiv}. Suppose that $\lambda^2 \asymp d_{\max}m\log m/n$, $d_{\max} \leq \mathcal{O}\{(n/\log m)^{1/2}\}$ and that $n \geq \mathcal{O}(\log m)$. Then, with probability greater than $1-2\alpha_0$, there exists a member of the population Markov equivalence class with associated parameter $\Gamma^\star_\mathrm{mec}$ such that i) $\|\hat{\Gamma}^\text{opt}-\Gamma^\star_\mathrm{mec}\|_F^2 \leq \mathcal{O}(\lambda^2{s}^\star)$ and ii) we recover the population Markov equivalence class; that is, $\mathrm{MEC}(\hat{\mathcal{G}}) = \mathrm{MEC}({\mathcal{G}}^\star)$ where $\hat{\mathcal{G}} = \mathcal{G}\{\hat{\Gamma}^\mathrm{opt}-\mathrm{diag}(\hat{\Gamma}^{\mathrm{opt}})\}$ is the estimated directed acyclic graph.
\label{thm:temp_2}
\end{theorem}
The proof of Theorem~\ref{thm:temp_2} follows a very similar argument as proof of Theorem~\ref{corr:equiv} (with $c=0$). 

Comparing Theorem~\ref{thm:temp_1} to Theorem~\ref{thm:main} and Theorem~\ref{thm:temp_2} to Theorem~\ref{corr:equiv}, we note that early stopping provides similar rates of convergence to population parameters. Thus, early stopping can result in significant computational speedups without sacrificing statistical accuracy (in terms of convergence rates).

\subsection{Example of outer approximation}
\label{oa_example}
In this section, we give a simple example to illustrate outer approximation. Consider the following integer programming problem: 
$$\min\limits_{x\in \mathbb{Z}_+} -2\log x + x,$$ which we know the optimal solution is $x=2$. Replacing $-2\log x$ with $y$, the outer approximation works as follows. Starting from any feasible point, say $x^{(1)}=4, y^{(1)}=-\infty$, we have $y^{(1)} < -2\log x^{(1)}$, so we need to add the first cutting plane: $y \geq -2\log 4 - 0.5(x- 4) = -0.5x + 2 - 2\log 4$. Then, by solving $$\min\limits_{x\in \mathbb{Z}_+, y\in \R} y + x \text{ s.t. } y \geq -0.5x + 2 - 2\log 4,$$ we obtain the solution $y^{(2)} = -1.27, x^{(2)} = 1$. Since $y^{(2)} < -2\log x^{(2)}$, we need to add another cutting plane to get a better lower bound. By solving $$\min\limits_{x\in \mathbb{Z}_+, y\in \R} y + x \text{ s.t. } y \geq -0.5x + 2 - 2\log 4, y \geq -2x+2,$$ we obtain the solution $y^{(3)} = -1.77, x^{(3)} = 2$. Since $y^{(3)} < -2\log x^{(3)}$, we add the third cutting plane: $y\geq  -x + 2-2\log 2$, which leads to the solution $y^{(2)} = -2\log 2 =-1.39, x^{(2)} = 2.$ Then the outer approximation algorithm stops.

\begin{figure}[t]
    \centering
    \begin{tikzpicture}[scale=0.8]
  \draw[->] (0,0) -- (6,0) node[right] {$x$};
  \draw[->] (0,-4) -- (0,3) node[above] {$y$};
  
  \draw[domain=0.25:6,smooth,variable=\x,blue] plot ({\x},{-2*ln(\x)+\x}) node[right] {$-2\log x + x$};
  \draw[domain=0.25:7.5,smooth,variable=\x,blue] plot ({\x},{-2*ln(\x)}) node[right] {$-2\log x$};
  \filldraw[black] (2,-1.38629436) circle (2pt) node[above right] {$(2, -1.39)$};
  \filldraw[black] (2,-1.7725887) circle (2pt) node[below left] {$(2, -1.77)$};
  \filldraw[black] (1,-1.2725887) circle (2pt) node[below left] {$(1, -1.27)$};
  
  \draw[dashed] (0,-0.7725887) -- (4,-2.77258) node[right] {first cutting plane};
  \draw[dashed] (0.1,1.8) -- (3,-4) node[right] {second cutting plane};
  \draw[dashed] node[above left] {third cutting plane}(0,0.6137056) -- (3.7,-3.08629436) ;
  
    \end{tikzpicture}
    \caption{Illustration of the outer approximation algorithm.}
    \label{fig:OA}
\end{figure}

As shown in Fig.~\ref{fig:OA}, the cutting planes are dynamically added to the problem, which offers two distinct advantages compared to the piece-wise linear approximation at predetermined breakpoints. First, we solve the problem exactly. Second, it is computationally more efficient by obviating the need for the a priori creation of a piece-wise linear approximation function with a large number of breakpoints
to reduce the approximation error as is done by \texttt{Gurobi}'s general constraint attribute.

\subsection{Comparison with other benchmarks using true moral graph as superstructure}\label{sec:compare_benchmarks}
In Table \ref{tab:compare_benchmarks}, we present results where true moral graphs are used as superstructures, if applicable. {We observe that the improvement of our algorithm over greedy equivalence search is not as significant in Table \ref{tab:compare_benchmarks} (where a true moral graph is supplied as input) as compared to Table     \ref{tab:compare_benchmarks_est}
 (where an estimated moral graph is supplied). To better understand this phenomenon, we recall that greedy equivalence search is a greedy heuristic that makes a locally optimal update at every iteration. When the input superstructure is constrained considerably (as is the case when the true moral graph is supplied as input), there is less opportunity for greedy equivalence search to make ``bad" updates. Nevertheless, as a true moral graph is rarely available in practice, we believe that the comparisons in Table \ref{tab:compare_benchmarks_est}
 are more practically relevant than those of Table \ref{tab:compare_benchmarks}.}

\begin{table}[b]
\def~{\hphantom{0}}
\caption{{Comparison with benchmarks using the true moral graph as superstructure, if applicable}}
    \centering
    \resizebox{\columnwidth}{!}{\begin{tabular}{lcccccccccccccccccc}
    
        & \multicolumn{2}{c}{\hdbu}  & \multicolumn{2}{c}{\td} & \multicolumn{3}{c}{\misocp} & \multicolumn{2}{c}{\ges} & \multicolumn{3}{c}{\micp}   \\
       Network.$m$.$s^\star$ & Time  & $d_{\mathrm{cpdag}}$ & Time  & $d_{\mathrm{cpdag}}$ & Time  & \rgap & $d_{\mathrm{cpdag}}$ & Time  & $d_{\mathrm{cpdag}}$ & Time  & \rgap & $d_{\mathrm{cpdag}}$  \\
       Dsep.6.6  &   $\leq 1$  & 9.2$\pm$4.4   & $\leq 1$   & 2.0$\pm$0.2 &  $\leq 1$  & * & 2.0$\pm$0 & $\leq 1$ & 1.9$\pm$0.4 & $\leq 1$  & * &2.0$\pm0$  \\
       Asia.8.8  &   $\leq 1$  & 20.7$\pm$4.8  & $\leq 1$   & 11.6$\pm$3.1 & $\leq 1$  & * & 10.9$\pm$2.0 & $\leq 1$ & 2.1$\pm$0.6  & $\leq 1$  & * & 2.1$\pm$0.4   \\
       Bowling.9.11  &   $\leq 1$  & 5.6$\pm$2.6 & $\leq 1$  & 7.0$\pm$2.1 & $\leq 1$  & * & 4.1$\pm$2.3 & $\leq 1$ & 2.4$\pm$1.2  & $\leq 1$  & * &2.0$\pm0$ \\
       InsSmall.15.25  &   $\leq 1$   & 39.9$\pm$7.1  & $\leq 1$   & 12.4$\pm$4.9 &  $5$ & *  & 8.0$\pm$0.0  & $\leq 1$ & 24.2$\pm$8.4  & T  & 0.05 & 7.6$\pm$2.2   \\
       Rain.14.18  &   $\leq 1$   & 16.7$\pm$3.8   & $\leq 1$  & 3.0$\pm$1.6  &  $1$ & *  & 2.0$\pm0$ & $\leq 1$ & 5.2$\pm$4.5   & $88$  & * &2.0$\pm0$ \\
       Cloud.16.19  &   $\leq 1$   & 43.8$\pm$7.8 & $\leq 1$   & 20.4$\pm$4.0   & $\leq 1$ & *  & 19.5$\pm$1.5  & $\leq 1$ & 3.0$\pm$1.4   & $11$  & * & 4.6$\pm$1.0   \\
       Funnel.18.18  &   $\leq 1$  & 15.9$\pm$3.0    & $\leq 1$   & 2.1$\pm$0.3   &  $\leq 1$  & * & 2.0$\pm$0.0 & $\leq 1$ & 3.3$\pm$3.9  & $62$  & * & 2.1$\pm$0.3  \\
       Galaxy.20.22  &   $\leq 1$   & 44.2$\pm$3.9 & $\leq 1$   & 27.9$\pm$6.1  &  $\leq 1$  & * & 6.2$\pm$5.2 & $\leq 1$ & 1.7$\pm$2.1   & $114$  & * &1.0$\pm0$  \\
       Insurance.27.52  &   $\leq 1$  & 43.8$\pm$6.6   & $\leq 1$  & 35.8$\pm$5.7 &  588   & * & 14.1$\pm$3.4 & $\leq 1$ & 25.8$\pm$10.4   & T  & 0.24 & 10.7$\pm$7.0    \\
       Factors.27.68 & $\leq 1$  & 26.1$\pm$5.8 & $\leq 1$   & 46.2$\pm$7.5  &T & 0.02 & 33.7$\pm$5.4 & $\leq 1$ & 63.7$\pm$7.7   & T  & 0.30 & 47.7$\pm$9.5 \\
       Hailfinder.56.66  &   4  & 132.6$\pm$19.5    & $\leq 1$  & 57.6$\pm$10.3   &  T  & 0.04 & 8.6$\pm$6.1  & $\leq 1$ & 13.3$\pm$9.1   & T  & 0.23 & 14.9$\pm$11.2    \\
       Hepar2.70.123  &   7   & 133.6$\pm$12.7   & 1   & 71.0$\pm$7.6  & T & 0.06  & 29.9$\pm$11.9  & $\leq 1$ & 37.7$\pm$13.2  & T  & 0.29 & 55.0$\pm$13.3  \\
    \end{tabular}}
    \caption*{
Here, \hdbu, high-dimensional bottom-up; TD, top-down; MISOCP, mixed-integer second-order cone program; \ges, greedy equivalence search;  $d_{\mathrm{cpdag}}$, differences between the true and estimated completed partially directed acyclic graphs; \rgap, relative optimality gap; T refers to reaching the time limit given by $50m$. All results are averaged across 30 independent trials.}    \label{tab:compare_benchmarks}
\end{table}


\subsection{Additional results}
\label{sec:Other_metrics}

{To further elucidate the comparison with greedy equivalence search in Table~4, we present a box plot of $d_{\mathrm{cpdag}}$ values in Fig.~\ref{fig:comp_ges_box}. We see that for most of the graphs, the upper quantile $d_{\mathrm{cpdag}}$ values of our method are below the lower quantile $d_{\mathrm{cpdag}}$ values of greedy equivalence search. Indeed, after applying the Wilcoxon signed-rank test to compare the mean $d_{\mathrm{cpdag}}$ of greedy equivalence search with the mean $d_{\mathrm{cpdag}}$, we find statistically significant results for all but three of the graphs (with a $p$-value less than 0.05); the $p$ values are shown in the box plot.
\begin{figure}
    \centering
    \includegraphics[width=1\linewidth]{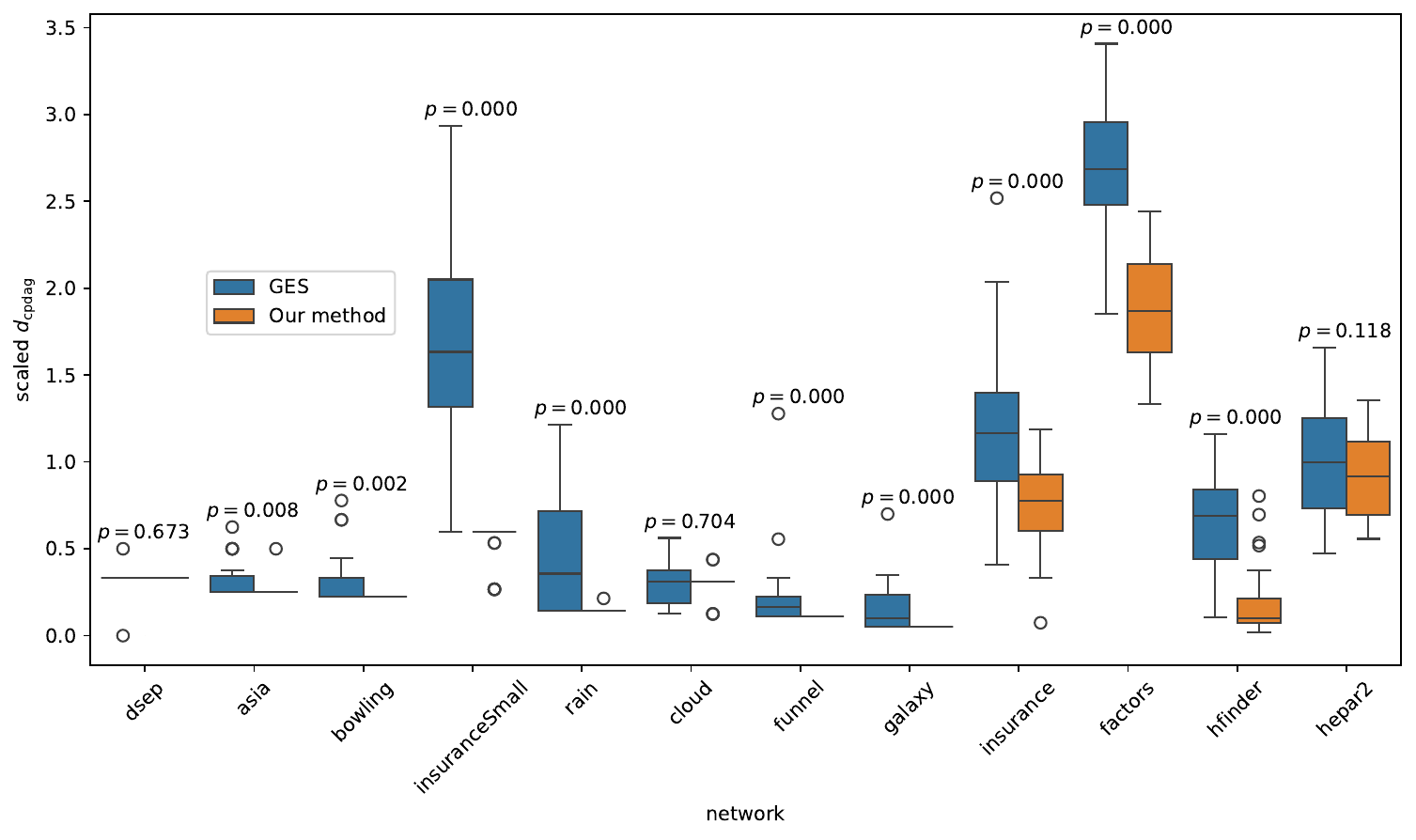}
    \caption{{Box plot for scaled $d_{\mathrm{cpdag}}$ for our method and greedy equivalence search with estimated moral graphs, and the $p$-values of the Wilcoxon signed-rank test.}}
    \label{fig:comp_ges_box}
\end{figure}

In Tables~\ref{tab:compare_other1} and \ref{tab:compare_other2}, we report the structural hamming distances of the undirected skeleton of the true directed acyclic graph and the corresponding skeleton of the estimated network, the true positive rate, and the false positive rate.

\begin{table}[ht]
\label{tab:other_metrics1}
\caption{{Structural hamming distance, true and also positive rate results,  part I}}
    \centering
    \scalebox{0.7}{{\begin{tabular}{cccccccccclll}
    
        & \multicolumn{3}{c}{\hdbu} & \multicolumn{3}{c}{\td} & \multicolumn{3}{c}{\misocp-True} & \multicolumn{3}{c}{\misocp-Est}\\
       Network($m$) & SHDs & TPR & FPR & SHDs & TPR & FPR& SHDs & TPR & FPR & SHDs & TPR &FPR\\
       Dsep(6)  &   3.9& 0.983&   0.126
&   1.0& 0.833&   0.001
&   1.0& 0.833&   0
& 1.0& 0.833&0
\\
       Asia(8)  &   11.8& 0.9&   0.197
&   4.9& 0.862&   0.067
&   2.6& 0.854&   0.223
& 2.7& 0.85&0.225
\\
       Bowling(9)  &   3.2& 0.918&   0.033
&   2.3& 0.909&   0.018
&   1.5& 0.909&   0.083
& 1.5& 0.909&0.083
\\
       InsSmall(15)  &   21.8& 0.968&   0.105
&   3.6& 0.948&   0.012
&   1.0& 0.96&   0.006
& 1.4& 0.953&0.015
\\
       Rain(14) &   9.2& 0.948&   0.047
&   1.4& 0.944&   0.002
&   1.0& 0.944&   0.037
& 1.0& 0.944&0.036
\\
       Cloud(16)   &   20.7& 0.93&   0.082
&   6.8& 0.925&   0.023
&   2.3& 0.933&   0.072
& 2.8& 0.919&0.073
\\
       Funnel(18)   &   10.0& 0.946&   0.03
&   1.1& 0.944&   0
&   1.0& 0.944&   0
& 1.0& 0.944&0
\\
       Galaxy(20)  &   24.1& 0.958&   0.061
&  10.0& 0.953&   0.024
&   2.4& 0.952&   0.022
& 2.4& 0.952&0.022
\\
       Insurance(27)  &   25.2& 0.986&   0.036
&   18.3& 0.946&   0.023
&   3.3& 0.963&   0.01
& 3.3& 0.962&0.011
\\
       Factors(27) & 11.1& 0.966&   0.013
& 21.4& 0.775&   0.009
&   15.9& 0.825&   0.033
& 18.4& 0.783&0.033
\\
       Hailfinder(56)  &   87.8& 0.973&   0.028
&   30.2& 0.922&   0.008
&   3.2& 0.975&   0.004
& 6.7& 0.961&0.009
\\
       Hepar2(70)   &   91.7& 0.983&   0.019&   40.9& 0.849&   0.005&   8.8& 0.967&   0.005& 12.6& 0.957&0.008\\
    \end{tabular}}}
\caption*{
\hdbu, high-dimensional bottom-up; \td, top-down; \misocp, mixed-integer second-order conic program; True, using the true moral graph as superstructure; Est, using estimated moral graph;SHDs, structural hamming distance of undirected graph skeletons; TPR, true positive rate; FPR, false positive rate; 
}    \label{tab:compare_other1}
\end{table}

\begin{table}[ht]
\label{tab:other_metrics2}
\caption{{Structural hamming distance, true and also positive rate results,  part II}}
    \centering
    \scalebox{0.7}{\begin{tabular}{ccccccccccccc}
        & \multicolumn{3}{c}{\ges-True} & \multicolumn{3}{c}{\micp-True}& \multicolumn{3}{c}{\ges-Est} & \multicolumn{3}{c}{\micp-Est}\\
       Network($m$) & SHDs & TPR & FPR & SHDs & TPR & FPR & SHDs & TPR & FPR  & SHDs & TPR & FPR \\
       Dsep(6)  &     1.0& 0.862&   0.032
&   1.0& 0.833&   0
&   1.0& 0.862&   0.033
& 1.0& 0.833& 0
\\
       Asia(8)  &    1.0& 0.815&   0.001
&   1.0& 0.875&   0
&   1.2& 0.8&   0.005
& 1.0& 0.875& 0
\\
       Bowling(9)  &     1.1& 0.909&   0.02
&   1.0& 0.909&   0
&   1.4& 0.906&   0.025
&  1.0& 0.909& 0
\\
       InsSmall(15) &   7.0& 0.553&   0.041
&   2.6& 0.899&   0
&   7.2& 0.547&   0.043
& 2.6& 0.896& 0
\\
       Rain(14) &   1.9& 0.867&   0.009
&   1.0& 0.944&   0
&   2.4& 0.83&   0.013
&  1.0& 0.944& 0
\\
       Cloud(16) &   1.0& 0.897&   0.001
&   1.9& 0.902&   0
&   1.9& 0.872&   0.006
& 2.0& 0.9& 0.001
\\
       Funnel(18) &   1.4& 0.894&   0.003
&   1.1& 0.944&   0
&   1.9& 0.888&   0.005
& 1.0& 0.944& 0
\\
       Galaxy(20) &   1.1& 0.949&   0.001
&   1.0& 0.955&   0
&   2.3& 0.943&   0.005
& 1.0& 0.955& 0
\\
       Insurance(27)  &   5.0& 0.756&   0.014
&  2.4& 0.969&   0.003
&   8.9& 0.73&   0.022
&5.0& 0.953& 0.009
\\
       Factors(27) & 14.6& 0.526&   0.029
& 18.5& 0.776&   0.012
& 22.4& 0.491&   0.036
& 21.5& 0.748& 0.015
\\
       Hailfinder(56)  &   4.2& 0.935&   0.003
&   4.7& 0.971&   0.002
&   14.7& 0.798&   0.006
&  3.6& 0.967& 0.001
\\
       Hepar2(70) &   11.7&   0.89& 0.005&  13.0&   0.956& 0.003&   33.0&   0.821& 0.01& 17.1& 0.945& 0.005\\
    \end{tabular}}
   \caption*{
SHDs, the structural hamming distance of undirected graph skeletons; TPR, true positive rate; FPR, false positive rate;}    \label{tab:compare_other2}
\end{table}

\subsection{Comparison with other benchmarks on non-Gaussian models}
\label{sec:non-Gaussian}
While our proposed method explicitly maximizes the Gaussian likelihood, some competing methods do not assume Gaussian errors. Therefore, in this section, we compare our method with other approaches when the errors are non-Gaussian. We follow the procedure from \cite{shimizu06a} to generate non-Gaussian samples. Specifically, we generate Gaussian errors first and subsequently pass them through a power non-linearity, which raises the absolute value of errors to an exponent in the interval $[0.5, 0.8]$ or $[1.2, 2.0]$, but keeping the original sign. From Table~\ref{tab:compare_benchmarks_non-Gaussian}, we see that even with non-Gaussian errors, our method outperforms others in most cases, especially for smaller instances that can be solved to optimality. Our method and greedy equivalence search \cite[]{chickering2002optimal} generally exhibit similar performance, as both approaches address the same objective functions, specifically log-likelihood functions with $\ell_0$-norm regularization. However, in certain instances, such as  InsSmall, our method achieves significantly fewer errors compared to greedy equivalence search. 
In Fig.~\ref{fig:synthetic_non-Gaussian}, we compare the methods on synthetic non-Gaussian datasets with $n=400$ and $m=10,15,20$. It can be seen that our method performs better than competing methods. Greedy equivalence search with true moral graphs as superstructures results in slightly larger errors and higher variances than ours, as it is a heuristic approach that does not guarantee optimal solutions.

\begin{table}[ht]
\def~{\hphantom{0}}
\caption{{Comparison with benchmarks on non-Gaussian models using true moral graph as superstructures if applicable}}
    \centering
    \resizebox{\columnwidth}{!}{\begin{tabular}{lcccccccccccccc}
    
        & \multicolumn{2}{c}{\hdbu} & \multicolumn{2}{c}{\td} & \multicolumn{3}{c}{\misocp}  & \multicolumn{2}{c}{\ges}  & \multicolumn{3}{c}{\micp} \\
       Network.$m$.$s^\star$ & Time  & $d_{\mathrm{cpdag}}$ & Time  & $d_{\mathrm{cpdag}}$ & Time  & \rgap & $d_{\mathrm{cpdag}}$ & Time & $d_{\mathrm{cpdag}}$ & Time  & \rgap & $d_{\mathrm{cpdag}}$ \\
        Dsep.6.6        & $\leq 1$ & 10.3$\pm4.1$   & $\leq 1$ & 2.1$\pm0.5$    & $\leq 1$    & *    & 2.0$\pm0$     & $\leq 1$ & 2.1$\pm0.5$   & $\leq 1$    & *    & 2.0$\pm0$     \\
        Asia.8.8        & $\leq 1$ & 22.2$\pm7.5$   & $\leq 1$ & 14.7$\pm2.0$   & $\leq 1$    & *    & 10.3$\pm1.9$  & $\leq 1$ & 2.1$\pm0.4$   & $1$    & *    & 2.5$\pm0.9$   \\
        Bowling.9.11     & $\leq 1$ & 7.8$\pm2.3$   & $\leq 1$ & 8.8$\pm2.1$    & $\leq 1$    & *    & 8.1$\pm2.4$   & $\leq 1$ & 2.4$\pm1.1$   & $2$         & *    & 2.0$\pm0$     \\
        InsSmall.15.25   & $\leq 1$ & 39.1$\pm5.5$  & $\leq 1$ & 16.6$\pm3.4$   & $3$         & *    & 10.3$\pm2.7$  & $\leq 1$ & 22.8$\pm7.1$  & T  & 0.04 & 8.1$\pm2.6$   \\
        Rain.14.18       & $\leq 1$ & 17.2$\pm4.6$  & $\leq 1$ & 4.3$\pm3.3$    & $\leq 1$    & *    & 2.0$\pm0$     & $\leq 1$ & 12.1$\pm8.6$  & $104$        & *    & 2.8$\pm2.4$     \\
        Cloud.16.19      & $\leq 1$ & 41.2$\pm7.5$  & $\leq 1$ & 20.5$\pm3.4$   & $\leq 1$    & *    & 20.6$\pm2.9$  & $\leq 1$ & 2.7$\pm1.6$   & $12$        & *    & 3.2$\pm1.5$   \\
        Funnel.18.18.     & $\leq 1$ & 17.1$\pm4.9$ & $\leq 1$ & 2.0$\pm0$      & $\leq 1$    & *    & 2.0$\pm0$   & $\leq 1$ & 3.0$\pm4.2$   & $117$        & *    & 2.2$\pm0.8$     \\
        Galaxy.20.22     & $\leq 1$ & 45.4$\pm5.9$   & $\leq 1$ & 31.2$\pm6.0$   & $1$    & *    & 19.3$\pm7.5$  & $\leq 1$ & 1.6$\pm2.1$   & $174$       & *    & 1.0$\pm0$     \\
        Insurance.27.52  & $\leq 1$ & 40.0$\pm7.4$   & $\leq 1$ & 39.7$\pm3.6$   & $403$       & *    & 22.2$\pm6.2$  & $\leq 1$ & 27.0$\pm9.0$  & T & 0.20 & 14.4$\pm8.3$ \\
        Factors.27.68    & $\leq 1$ & 28.0$\pm4.9$   & $\leq 1$ & 46.8$\pm8.5$   & $667$       & *    & 38.8$\pm5.4$  & $\leq 1$ & 64.2$\pm9.6$  & T & 0.24 & 51.1$\pm10.6$ \\
        Hailfinder.56.66 & 5        & 135.7$\pm20.0$ & $\leq 1$ & 75.2$\pm9.8$  & T & 0.06 & 42.7$\pm13.3$ & $\leq 1$ & 6.1$\pm7.1$   & T & 0.17 & 27.4$\pm15.3$ \\
        Hepar2.70.123     & 10        & 143.8$\pm16.2$ & 2.3        & 95.4$\pm11.4$ & T & 0.07 & 46.5$\pm12.1$ & $\leq 1$ & 33.2$\pm12.6$ & T & 0.24 & 58.9$\pm13.9$
    \end{tabular}}
    \caption*{
\hdbu, high-dimensional bottom-up; \td, top-down; \misocp, mixed-integer second-order cone program; \ges, greedy equivalence search;  $d_{\mathrm{cpdag}}$, differences between the true and estimated completed partially directed acyclic graphs; \rgap, relative optimality gap;}
    \label{tab:compare_benchmarks_non-Gaussian}
\end{table}

\begin{figure}[ht]
    \centering
    \subfloat[$m = 10$]{
        \includegraphics[scale=0.33]{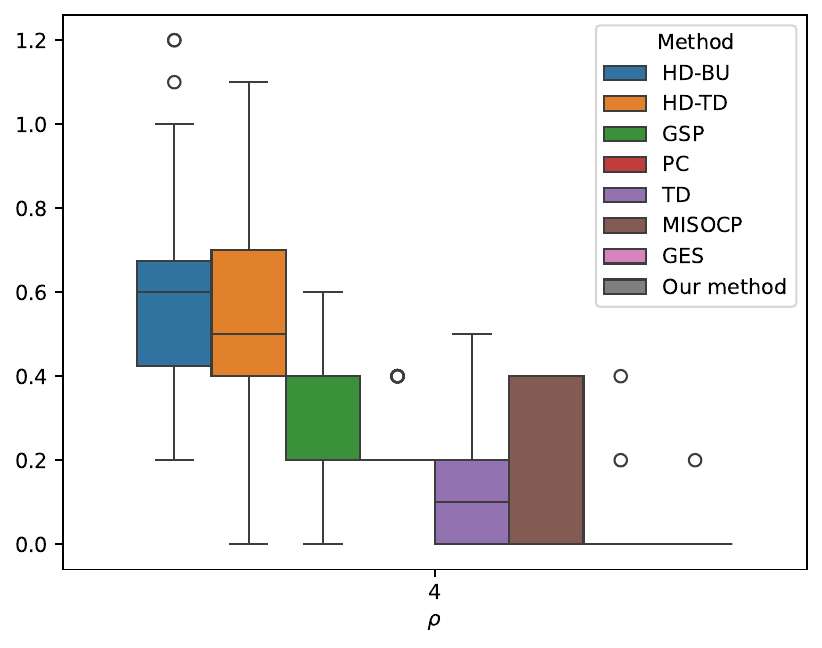}}
    \subfloat[$m=15$]{
        \includegraphics[scale=0.33]{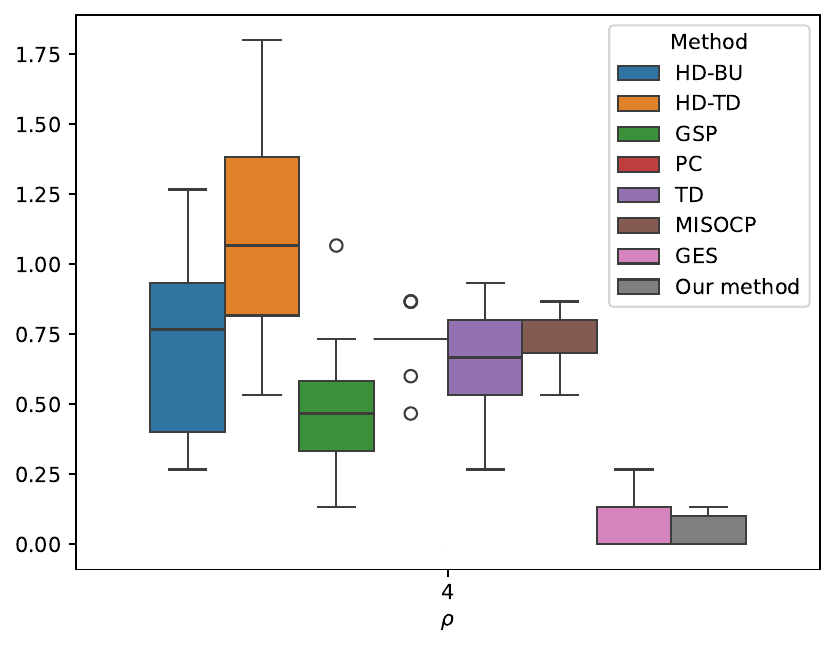}}
   \subfloat[$m=20$]{
        \includegraphics[scale=0.33]{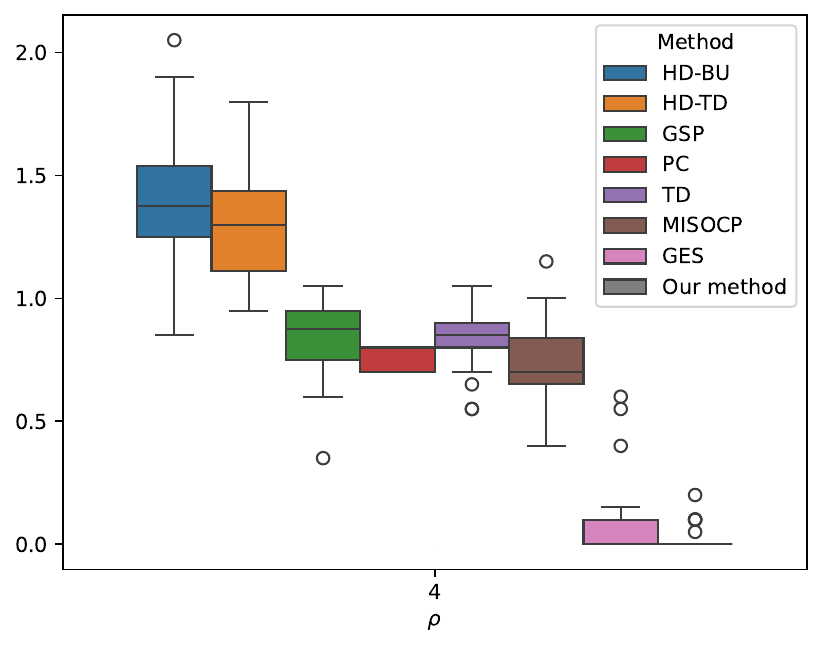}}
    \caption{{Box plots of scaled $d_{\mathrm{cpdag}}$ for our method and benchmarks with $\rho = 4$ and number of nodes $m=10, 15, 20$,  respectively. Here, $d_{\mathrm{cpdag}}$ values are scaled by the total number of edges in the true underlying directed acyclic graph. \hdbu, high-dimensional bottom-up; \hdtd, high-dimensional top-down; \gsp, greedy sparsest permutations; \pc, PC algorithm; \td, top-down;   \misocp, mixed-integer second-order conic program; \ges, greedy equivalence search; $d_{\mathrm{cpdag}}$, differences between true and estimated completed partially directed acyclic graphs with smaller values being better. 
 } }
    \label{fig:synthetic_non-Gaussian}

\end{figure}

\subsection{Additional results of the PC algorithm and the greedy sparsest permutation algorithm}\label{app:pc-gsp}

For completeness, Table~\ref{tab:compare_benchmarks_pc_gsp} summarizes our experiments comparing the $d_{\mathrm{cpdag}}$ and running time of the PC algorithm, the greedy sparsest permutation algorithm, and our method. Since these competing algorithms are outperformed by the greedy equivalence search algorithm, they were not included in the main table.


\begin{table}[ht]
\def~{\hphantom{0}}
\caption{{The PC algorithm and the greedy sparsest permutation algorithm using estimated superstructure}}
    \centering
    \scalebox{0.7}{\begin{tabular}{lccccccc}
    
         & \multicolumn{2}{c}{\pc}& \multicolumn{2}{c}{\gsp} & \multicolumn{3}{c}{\micp}\\
       Network.$m$.$s^\star$  & Time & $d_{\mathrm{cpdag}}$ & Time & $d_{\mathrm{cpdag}}$ & Time  & \rgap & $d_{\mathrm{cpdag}}$\\
       Dsep.6.6  & $\leq 1$ & 5.0$\pm0$  & $\leq 1$ & 2.0$\pm0$ & $\leq 1$  & * &2.0$\pm0$ \\
       Asia.8.8  & $\leq 1$ & 9.1$\pm0.4$ & $\leq 1$ & 5.2$\pm$2.1 & $\leq 1$  & * &2.1$\pm0.4$ \\
       Bowling.9.11  & $\leq 1$ & 14.7$\pm1.1$ & $\leq 1$ & 5.8$\pm$3.2 & $2$  & * &2.0$\pm0$ \\
       InsSmall.15.25  & $\leq 1$ & 34.6$\pm$1.3   & $\leq 1$ & 18.3$\pm$5.1 & T  & 0.05 &8.1$\pm1.9$ \\
       Rain.14.18  & $\leq 1$ & 17.1$\pm$1.8   & $\leq 1$ & 16.0$\pm$3.8 & $104$  & * &2.0$\pm0$ \\
       Cloud.16.19  & $\leq 1$ & 23.2$\pm$3.3  & $\leq 1$ & 14.0$\pm$2.0 & $28$  & * &4.8$\pm1.1$ \\
       Funnel.18.18  & $\leq 1$ & 21.3$\pm$0.7  & $\leq 1$ & 12.4$\pm$3.4 & $42$  & * &2.0$\pm0$ \\
       Galaxy.20.22  & $\leq 1$ & 29.7$\pm$1.0   & $\leq 1$ & 15.2$\pm$4.8   & $200$  & * &1.0$\pm0$ \\
       Insurance.27.52  & $\leq 1$ & 72.9$\pm$3.0   & $\leq 1$ & 37.9$\pm$7.0  & T  & 0.28 &19.8$\pm7.5$ \\
       Factors.27.68 & $\leq 1$ & 94.6$\pm$3.1   & $\leq 1$ & 54.2$\pm$7.4  & T  & 0.29 &51.1$\pm8.5$ \\
       Hailfinder.56.66  & $46$ & 91.7$\pm$5.2  & $\leq 1$ & 113.4$\pm$9.8  & T  & 0.19 &10.7$\pm11.4$ \\
       Hepar2.70.123  & 2 & 138.2$\pm$6.2& $\leq 1$ & 66.6$\pm$8.9  & T  & 0.32 &64.2$\pm17.5$ \\
    \end{tabular}}
\caption*{ \pc, PC algorithm;  GSP, greedy sparsest permutations; $d_{\mathrm{cpdag}}$, differences between the true and estimated completed partially directed acyclic graphs; \rgap, relative optimality gap; T refers to reaching the time limit given by $50m$.}    \label{tab:compare_benchmarks_pc_gsp}
\end{table}

\end{document}